\documentclass[11pt]{article}%
\usepackage{amsfonts}
\usepackage{fancyhdr}
\usepackage{tabularx}
\usepackage{mathtools}
\usepackage[colorlinks = True]{hyperref}
\hypersetup{
    colorlinks = true,
    linkcolor = blue,
    anchorcolor = blue,
    citecolor = blue,
    filecolor = blue,
    urlcolor = blue
    }
\usepackage[a4paper, top=2.5cm, bottom=2.5cm, left=2.2cm, right=2.2cm]%
{geometry}
\usepackage{times}
\usepackage[table,xcdraw]{xcolor}
\usepackage{amsmath}
\usepackage{amsthm}
\usepackage[capitalize, nameinlink]{cleveref}
\usepackage{changepage}
\usepackage{enumitem}
\usepackage{mathrsfs}
\usepackage{tcolorbox}
\usepackage{stackrel,amssymb}
\usepackage{graphicx}

\usepackage{appendix}
\usepackage[backend=biber,
style=alphabetic,
maxnames=50
]{biblatex}
\newtheorem{theorem}{Theorem}
\newtheorem{corollary}[theorem]{Corollary}

\newtheorem{definition}[theorem]{Definition}
\newtheorem{lemma}[theorem]{Lemma}
\newtheorem{remark}[theorem]{Remark}
\newtheorem{example}[theorem]{Example}
\newtheorem{proposition}[theorem]{Proposition}
\newtheorem{problem}[theorem]{Problem}

\newtheorem{observation}[theorem]{Observation}
\numberwithin{theorem}{subsection}

\newcommand{\game}{\textsc{Game}}
\newcommand{\prob}{\mathbf{P}}
\newcommand{\boldv}{\mathbf{v}}
\newcommand{\opt}{\mathsf{OPT}}
\newcommand{\median}{\mathbf{Med}}
\newcommand{\expect}{\mathbf{E}}

\newcommand{\classqm}{\mathcal{Q}(q,[m])}
\def\hmath$#1${\texorpdfstring{{\rmfamily\textit{#1}}}{#1}}

\newcommand{\core}{{\mathsf{core}}}
\newcommand{\gammacore}{\gamma\text{-}\mathsf{core}}
\newcommand{\citycore}{{\mathsf{citycore}}}
\newcommand{\gammacitycore}{\gamma\text{-}\mathsf{citycore}}

\author{Kiril Bangachev\thanks{Dept. of EECS, MIT. \texttt{kirilb@mit.edu} } \ and S. Matthew Weinberg\thanks{Dept. of CS, Princeton University. \texttt{smweinberg@princeton.edu}}}
\title{$q$-Partitioning Valuations: Exploring the Space Between\\
Subadditive and Fractionally Subadditive Valuations}
\begin{document}
\pagenumbering{roman}
\maketitle
\begin{abstract}

For a set $M$ of $m$ elements, we define a decreasing chain of classes of normalized monotone-increasing valuation functions from $2^M$ to $\mathbb{R}_{\geq 0}$, parameterized by an integer $q \in [2,m]$. For a given $q$, we refer to the class as \emph{$q$-partitioning}. A valuation function is subadditive if and only if it is $2$-partitioning, and fractionally subadditive if and only if it is $m$-partitioning. Thus, our chain establishes an interpolation between subadditive and fractionally subadditive valuations. We show that this interpolation is smooth ($q$-partitioning valuations are ``nearly'' $(q-1)$-partitioning in a precise sense,~\cref{thm:smoothness}), interpretable (the definition arises by analyzing the core of a cost-sharing game, \`{a} la the Bondareva-Shapley Theorem for fractionally subadditive valuations,~\cref{section:costsharing}), and non-trivial (the class of $q$-partitioning valuations is distinct for all $q$,~\cref{prop:exisetnce}).

We interpolate prior results that separate subadditive and fractionally subadditive for all\linebreak $q \in \{2,\ldots, m\}$. Two highlights are the following:
\begin{enumerate}
\item[i)] An $\Omega \left(\frac{\log \log q}{\log \log m}\right)$-competitive posted price mechanism for $q$-partitioning valuations. Note that this matches asymptotically the state-of-the-art for both subadditive ($q=2$)~\cite{DuttingKL20}, and fractionally subadditive ($q=m$)~\cite{FeldmanGL15}.
\item[ii)]  Two upper-tail concentration inequalities on $1$-Lipschitz, $q$-partitioning valuations over independent items. One extends the state-of-the-art for $q=m$ to $q<m$, the other improves the state-of-the-art for $q=2$ for $q > 2$. Our concentration inequalities imply several corollaries that interpolate between subadditive and fractionally subadditive, for example: $\mathbb{E}[v(S)]\le (1 + 1/\log q)\text{Median}[v(S)] + O(\log q)$. To prove this, we develop a new isoperimetric inequality using Talagrand's method of control by $q$ points, which may be of independent interest.
\end{enumerate}

We also discuss other probabilistic inequalities and game-theoretic applications of $q$-partitioning valuations, and connections to subadditive MPH-$k$ valuations~\cite{EzraFNTW19}.
\end{abstract}

\addtocounter{page}{-1}
\newpage
\section{Introduction}\label{sec:intro}
\pagenumbering{arabic}
\label{section:intro}
\subsection{Motivation}
Functions of the form $f:2^M\longrightarrow \mathbb{R}$ are a fundamental object of study in the fields of Algorithmic Game Theory and Combinatorial Optimization. For example, when $M$ is a set of items in an auction, $f(S)$ could indicate the value that an agent obtains from receiving the bundle $S$ (see more about combinatorial auctions in \cite[Chapter 11]{AGTbook}). When $M$ is a set of agents, $f(S)$ could indicate the cost that agents $S$ need to pay in order to purchase a given service together (see more about cost sharing in \cite[Chapter 15]{AGTbook}).


As set functions\footnote{We will use the terms \textit{set function} and \textit{valuation} interchangeably in the rest of the paper. This convention is motivated by the setting of combinatorial auctions in which $f(S)$ indicates the ``value'' of the subset of items $S.$} are motivated by real world processes --- auctions, cost sharing, and job scheduling among others --- the mathematical study of such functions usually assumes that they satisfy certain natural properties. Throughout the paper, we assume that all valuations satisfy the following two simple technical properties: they are \textit{monotone} ($f(S)\le f(T)$ whenever $S\subseteq T$) and \textit{normalized} ($f(\emptyset) = 0$). Economic considerations give rise to more complex conditions on set functions.
For example, frequently imposed is the condition of diminishing marginal values, also known as \textit{submodularity}. Another condition motivated by economics is complement-freeness in the values that an agent obtains from bundles of items, which is also known as \textit{subadditivity}. Finally, one could be interested in the existence of prices which incentivize cooperation among agents when purchasing a given service; this turns out to be equivalent to the \textit{fractionally subadditive} property (see \cref{section:costsharing}).\footnote{In this paper, the terms ``submodularity'', ``fractional subadditivity'', and ``subadditivity'' imply monotonicity and normalization.}



In this paper, we focus our attention on fractionally subadditive and subadditive set functions. Trivially, fractionally subadditive functions are a smaller class strictly contained in the class of subadditive functions. Something stronger turns out to be true --- Bhawalkar and Roughgarden show the existence of subadditive functions which are very far from being fractionally subadditive in a precise quantitative sense \cite{BhawalkarR11}. This difference between fractionally subadditive and subadditive valuations is not purely theoretical and has important implications. For example, in the context of combinatorial auctions, there exists a posted price mechanism that gives a $(1/2)$-approximation to the optimal welfare when all players have fractionally subadditive valuations~\cite{FeldmanGL15}, but the best known approximation ratio for subadditive valuations is $\Omega(\frac{1}{\log \log m}),$ where $m$ is the number of items \cite{DuttingKL20} (moreover, the~\cite{FeldmanGL15} framework providing a $(1/2)$-approximation for XOS provably cannot beat $O(\log m)$ for subadditive, and the~\cite{DuttingKL20} framework provably cannot beat $O(\log \log m)$). Similarly, in the context of concentration inequalities, a fractionally subadditive valuation $v$ has $\expect[v]$-subgaussian lower tails (see \cite[Corollary 3.2]{Vondrak10}), but such a strong dimension-free concentration provably does not hold for subadditive valuations (see \cite[Section 4]{Vondrak10}).

What if a set function is ``somewhere in between being subadditive and being fractionally subadditive''? On the one hand, as it is not fractionally subadditive, one cannot use the strong guarantees of fractional subadditivity (such as in posted price mechanisms or subgaussian concentration) when analyzing it. On the other hand, as the set function could be significantly more structured than an arbitrary subadditive function, it is perhaps inefficient to simply use the much weaker properties guaranteed by subadditivity (especially, those that provably cannot be improved for all subadditive functions). In this paper, we construct a smooth interpolation between fractional subadditivity and subadditivity. Explicitly, we define a chain of function classes that starts with fractionally subadditive set functions and expands to subadditive set functions. Our goal is to understand how the behaviour of these function classes changes along the chain. We focus on several setups in which subadditive and fractionally subadditive valuations have received significant attention in the literature, and in which strong claims for fractionally subadditive valuations provably don't hold for all subadditive valuations.

\subsection{Results Part I: Defining \hmath$q$-partitioning valuations}

Our chain of classes is parametrized by a positive integer parameter $q$ ranging between $q= |M|$ (which corresponds to the fractionally subadditive case) and $q = 2$ (which corresponds to the subadditive case). The number $q$ corresponds to the complexity of fractional covers under which the valuation function is non-diminishing. We call the respective classes \textit{$q$-partitioning} and the resulting interpolation the \textit{partitioning interpolation.} We give a formal definition in \cref{def:qpartprimal}. We then establish that the partitioning interpolation satisfies several desirable properties:
\setlist{nolistsep}
\begin{itemize}[noitemsep]
    \item \textbf{Interpretability:} In \cref{section:costsharing}, we present an economic interpretation of $q$-partitioning via the core of a cost-sharing game \`{a} la the Bondareva-Shapley theorem which characterizes fractionally subadditive valuations~\cite{Bondareva63,Shapley67} --- See also~\cite[Theorem~15.6]{AGTbook}. In slightly more detail, say there is a service that can be acquired by set $S$ of players if they together pay $c(S)$. One can then ask, for any subset $T \subseteq [m]$ whether or not there exist non-negative prices $\{p_i\}_{i \in T}$ such that: a) $\sum_{i \in T} p_i = c(T)$ (service is purchased for $T$) and b) for all $S\subseteq T$, $\sum_{i \in S} p_i \leq c(S)$ (no set $S\subseteq T$ wishes to deviate and purchase the service just for themselves). The Bondareva-Shapley theorem, applied to monotone normalizd cost functions, states that such prices exist for all $T$ \emph{if and only if $c(\cdot)$ is fractionally subadditive}. 

    Consider instead modifying the game so that players are grouped into $q$ fully-cooperative \emph{cities} (that is, cities will always act as a coherent unit, and will act in the best interest of the entire city). One can then ask, for any subset $T \subseteq [m]$ and any partitioning of $T$ into $q$ cities $T_1,\ldots, T_q$, do there exist non-negative prices $\{p_i\}_{i \in [q]}$ such that: a) $\sum_{i \in [q]} p_i = c(T)$ (service is purchased for $T$) and b) for all $S \subseteq [q]$, $\sum_{i \in S} p_i \leq c(\cup_{i \in S} T_i)$ (no set $S$ of cities wishes to deviate and purchase the service just for themselves). \cref{prop:maincostsharing} establishes that such prices exist for all $T$ and all partitionings of $T$ into at most $q$ cities \emph{if and only if $c(\cdot)$ is $q$-partitioning}.
    \item \textbf{Smoothness of The Interpolation:} In \cref{thm:smoothness}, we show that our chain of classes is smooth in the sense that every $q$-partitioning valuation is almost $(q+1)$-partitioning. Formally, \cref{thm:smoothness} establishes that the class of $q$-partitioning valuations is $(1-1/q)$-close to the class of $(q+1)$-partitioning valuations. We provide a formal definition of closeness in \cref{def:closeness}, but note briefly here that it is the natural extension of closeness to XOS valuation functions from~\cite{BhawalkarR11} extended to $q < m$.
    
    \item \textbf{Existence of Classes:} In \cref{prop:exisetnce}, we show that for each $m = |M|$ and $2\le q \le m,$ there exist $q$-partitioning valuations over $M$ that are not $(q+1)$-partitioning. In other words, none of the $m-1$ classes ``collapses'' to a lower level.
\end{itemize}

\subsection{Results Part II: Posted price mechanisms and concentration inequalities}

Our main results apply the partitioning interpolation to two canonical problems where subadditive and fractionally subadditive valuations are ``far apart.'' Our main results provide analyses that smoothly degrade from fractionally subadditive to subadditve as $q$ decreases -- this enables stronger guarantees for wide classes of structured subadditive functions which (provably) cannot be obtained for all subadditive functions.\\

\noindent\textbf{Posted Price Mechanisms.} Posted price mechanisms are a core objective of study within Algorithmic Game Theory, including multi-dimensional mechanism design~\cite{ChawlaHMS10}, single-dimensional mechanism design~\cite{Yan11,AlaeiHNPY15}, and the price of anarchy~\cite{FeldmanGL15, DuttingFKL20}. Posted price mechanisms list a price $p_i$ for each item $i \in [m]$, then visit the bidders one at a time and offer them to purchase any remaining set $S$ of items at price $\sum_{i \in S} p_i$ (and these items become unavailable for all future bidders). Of course, strategic players will pick the remaining set $S$ that maximizes $v_i(S) - \sum_{i \in S} p_i$. 

Of key importance to multiple of these agendas is the following basic question: to what extent can posted price mechanisms optimize welfare in Bayesian settings? Specifically, assume that each bidder $i$'s valuation function $v_i(\cdot)$ is drawn independently from a known distribution $D_i$ over valuations in some class $\mathcal{V}$. The optimal expected welfare is simply $\mathbb{E}_{\vec{v}\leftarrow \times_i D_i}[\max_{\text{partitions } S_1,\ldots, S_n}\{\sum_i v_i(S_i)\}]$. When strategic players participate in a posted-price mechanism with prices $\vec{p}$, some other partition of items is selected, guaranteeing some other expected welfare. What is the maximum number $\alpha(\mathcal{V})$ such that for all $D = \times_i D_i$ supported on $\mathcal{V}$, there exists a posted-price mechanism that results in expected welfare at least an $\alpha(\mathcal{V})$-fraction of the optimal welfare? Besides being the main question of study in works such as~\cite{FeldmanGL15,DuttingFKL20,DuttingKL20}, resolving this question has downstream implications for revenue-maximization in multi-dimensional settings due to~\cite{CaiZ17}. 

For the class of fractionally subadditive valuations,~\cite{FeldmanGL15} establish a $1/2$-approximation, which also implies a $1/\log_2(m)$-approximation for subadditive valuations. However, their techniques provably cannot yield stronger guarantees for subadditive valuations~\cite{BhawalkarR11,DuttingFKL20}. Recent breakthrough work of~\cite{DuttingKL20} designs a new framework for subadditive valuations that yields an $\Omega(1/\log_2\log_2(m))$-approximation, but aspects of their framework also provably cannot provide stronger guarantees. In this sense, there is a strong separation between the state-of-the-art guarantees on posted price mechanisms for fractionally subadditive and subadditive valuations (and also, there is a permanent separation between what can be achieved within the aforementioned frameworks).\\

\noindent\textbf{Main Result I:} Our first main result provides an $\Omega(\frac{\log \log q}{\log \log m})$-competitive posted price mechanism when all distributions are supported on $q$-partitioning valuations. This is stated in \cref{thm:postedpriceqpart}.\\

\noindent Note that this guarantee matches both the constant factor approximation in the fractionally subadditive case (setting $q = m$) and the $\Omega(\frac{1}{\log \log m})$ factor in the subadditive case (setting $q = 2$) and interpolates between the two approximation factors when $q$ is in between. In particular, note that this matches the state-of-the-art in both extremes, and matches the best guarantees achievable by the~\cite{DuttingKL20} approach in both extremes.\\

\noindent\textbf{Concentration Inequalities.} Consider a function $f$, and a set $S$ selected by randomly including each item $i$ independently (not necessarily with the same probability). It is often of interest to provide upper tail bounds on the distribution of $f(S)$ compared to $\expect[f(S)]$. McDiarmid's inequality is one such example when $f$ is $1$-Lipschitz.\footnote{$f(\cdot)$ is $1$-Lipschitz if $|f(S) - f(S \cup \{i\}| \leq 1$ for all $S,i$.} It is further the case that when $f(\cdot)$ is subadditive or fractionally subadditive, even stronger upper tail bounds are possible~\cite{Vondrak10}.


For example, if $f$ is both $1$-Lipschitz and subadditive, Schechtman's inequality implies that the probability that $f(S)$ exceeds twice its median plus $x$ decays exponentially in $x$~\cite{schectman}.\footnote{Schectman's inequality is more general than this, but this is one common implication. See \cref{eq:schehtman} for the general statement.} Importantly, Schectman's inequality provably cannot ``kick in'' arbitrarily close to the median~\cite{Vondrak10}.\\

\noindent\textbf{Main Result II:} \cref{thm:qparttailspecial} improves Schectman's inequality across the partitioning interpolation. In particular, our improvement implies that for all $1$-Lipschitz and $q$-partitioning $f$, the probability that $f([m])$ exceeds $(1+\log_2(q))$ times its median plus $x$ decays exponentially in $x$. This is stated in \cref{thm:qparttailspecial} . In particular, \cref{thm:qparttailspecial}  makes use of a new isoperimetric inequality that may be of independent interest, and that is stated in \cref{thm:talagrandgenerals}.\footnote{Note that this result, and that of~\cite{schectman} applies in a more general setting where there is a collection of independent random variables $X_1,\ldots, X_m$, that parameterize a function $f_{\vec{X}}: 2^{[m]}\rightarrow \mathbb{R}$, which is subadditive for all $\vec{X}$. Like~\cite{Vondrak10}, we provide proofs in the canonical setting referenced in the text for simplicity of exposition.}\\

Similarly, if $f$ is both $1$-Lipschitz and fractionally subadditive,~\cite{Vondrak10} establishes that $f$ is self-bounding. \cite{BoucheronLM00} establish ``Chernoff-Bernstein-like'' concentration inequalities on self-bounding functions, which in particular imply that $f([m])$ has $\expect[f([m])]$-subgaussian lower tails and slightly weaker upper tails.\footnote{A random variable $X$ is $\sigma^2$-subgaussian if the following inequality holds.  The log-moment generating function defined by $\psi_X(\lambda):=\log \expect[\exp(\lambda(X - \expect[X]))]$ exists for all real numbers $t$ and, furthermore, satisfies 
$\psi_X(\lambda)\le \frac{\lambda^2\sigma^2}{2}.$ It is well known that if $X$ is $\sigma^2$-subgaussian, then
$\prob[X\ge \expect[X] + t]\le \exp(-\frac{t^2\sigma^2}{2})$ and 
$\prob[X\le \expect[X] - t]\le \exp(-\frac{t^2\sigma^2}{2}).$
}\\

\noindent\textbf{Main Result III:} \cref{thm:selfboundingqpart} extends this across the partitioning interpolation. Specifically, our result establishes that for all $1$-Lipschitz and $q$-partitioning $f$, $f([m])$ is $(\lceil m/q \rceil,0)$-self bounding, which implies by~\cite{McDiarmidR06,BoucheronLM09} that $f([m])$ has a $\lceil m/q \rceil\cdot \expect[f([m])]$-subgaussian lower tail and a slightly worse Bernstein-like upper tail.\\

It is worth noting that~\cite{schectman}, based on Talagrand's method of control by $q$ points, is the state-of-the-art for concentration of subadditive functions, while~\cite{Vondrak10}, based on the method of self-bounding functions~\cite{BoucheronLM00,McDiarmidR06,BoucheronLM09}, is state-of-the-art for fractionally subadditive functions. Our main results extend both across the partitioning interpolation, but neither of the two approaches yields ``tight'' results at both ends --- our extension of~\cite{schectman} gives sharper results for small $q$, and our extension of~\cite{Vondrak10} gives sharper results for larger $q$. This is to be expected, as the two methods are genuinely distinct. 

\subsection{Related Work and Connection to Subadditive MPH-\hmath$k$}

\noindent\textbf{Hierarchies of Valuation Functions.} Prior to our work, there has been significant interest in exploring the space of valuation functions with \emph{parameterized complementarities}~\cite{AbrahamBDR12,FeigeFIILS15,FeigeI13,FeldmanFMR16,FeldmanI14,FeldmanI17, EdenFFTW21}. That is, the simplest level of the hierarchy is (fractionally) subadditive valuations, the second level of the hierarchy already contains functions that are not subadditive, and the final level of the hierarchy contains all monotone functions. These works are distinct from ours in that they explore the space between (fractionally) subadditive valuations and arbitrary monotone valuations, whereas our work explores the space between fractionally subadditive and subadditive valuations. 

To the best of our knowledge, the only prior work exploring the space between fractionally subadditive and subadditive valuations is~\cite{EzraFNTW19}. Their main results concern the communication complexity of two-player combinatorial auctions for subadditive valuations, but they also provide improved parameterized guarantees for valuations that are subadditive and also MPH-$k$~\cite{FeigeFIILS15}. A detailed comparison to our work is therefore merited:
\begin{itemize}
\item The partitioning interpolation follows from a first-principles definition (\cref{section:costsharing}). On the other hand, the MPH hierarchy explores the space between fractionally subadditive and arbitrary monotone valuations, and~\cite{EzraFNTW19} restrict attention to the portion of this space that is also subadditive.
\item Our main results consider posted price mechanisms and concentration inequalities, neither of which are studied in~\cite{EzraFNTW19}.~\cite{EzraFNTW19} study the communication complexity of combinatorial auctions (where the gap between fractionally subadditive and subadditive is only constant), which is not studied in our work.
\item We show (\cref{lem:qpartandmphk}) that all $q$-partitioning valuations are also MPH-$\lceil m/q\rceil$. Therefore, we can conclude a $(1/2 + 1/\log_2 (\lceil m/q \rceil))$-approximation algorithm for two-player combinatorial auctions with $q$-partitioning valuations using~\cite{EzraFNTW19} (this is the only result of their paper concerning functions between fractionally subadditive and subadditive).
\item We further show that $q$-partitioning admits a dual definition (\cref{def:qpartdual}, similar to the duality between XOS and fractionally subadditive). A particular feasible dual solution implies a witness that $q$-partitioning valuations are MPH-$\lceil m/q\rceil$. This suggests that our dual definition is perhaps ``the right'' modification of subadditive MPH-$k$ so that a dual definition exists.
\end{itemize}
\vspace{2mm}
\noindent\textbf{Posted price mechanisms.} Posted price mechanisms are a core object of study within Algorithmic Game Theory. Variants of posted price mechanisms achieve state-of-the-art guarantees for wide ranges of combinatorial auctions~\cite{AssadiKS21,DobzinskiNS12}. Posted price mechanisms are strongly obviously strategyproof~\cite{Li17,PyciaT19}. Posted price mechanisms have also been used in Bayesian settings to study the price of anarchy for welfare~\cite{FeldmanGL15,DuttingFKL20,DuttingKL20}, revenue maximization in multi-dimensional settings~\cite{ChawlaHMS10,KleinbergW19,ChawlaM16, CaiZ17}, and revenue maximization in single-dimensional settings~\cite{Yan11,AlaeiHNPY15,FengHL19, JinLQTX19,JinLTX19,JinJLZ21}. Most relevant to our work is the study of posted price mechanisms in Bayesian settings for welfare, where the state-of-the-art is a $(1/2)$-approximation for fractionally subadditive valuations~\cite{FeldmanGL15}, and a $\Omega(1/\log_2\log_2(m))$-approximation for subadditive valuations~\cite{DuttingKL20}. These results further imply approximation guarantees of the same asymptotics for multi-dimensional mechanism design via~\cite{CaiZ17}, and it is considered a major open problem whether improved guarantees are possible for subadditive valuations. Our work provide improved guarantees across the partitioning interpolation (of $\Omega(\log_2\log_2 (q)/\log_2\log_2(m))$), which matches the state-of-the-art at both endpoints (and moreover, is provably tight at both endpoints for the approach of~\cite{DuttingKL20}).\\ 

\noindent\textbf{Concentration Inequalities.} Concentration inequalities on functions of independent random variables are a core tool across many branches of Computer Science. For example, they are widely used in Bayesian mechanism design~\cite{RubinsteinW18,ChawlaM16, CaiZ17, KothariMSSW19}, learning theory
\cite{BalcanH11,FeldmanV13}, and discrete optimization \cite{FairsteinKS21}. Vond\'{a}k's wonderful note on concentration inequalities of this form gives the state-of-the-art when $f$ is fractionally subadditive and subadditive, and mentions other applications~\cite{Vondrak10}. Our results extend both the state-of-the-art for subadditive and fractionally subadditive across the partitioning interpolation. In addition, we provide a new isoperimetric inequality based on Talagrand's method of control by $q$ points.

\subsection{Summary and Roadmap}
\cref{sec:prelim} immediately follows with formal definitions. \cref{sec:defineqpart} defines the partitioning interpolation, and provides several basic properties (including an interpretation via cost-sharing, and a dual formulation). \cref{section:postedprices} overviews our first main result: an $\Omega(\frac{\log\log q}{\log \log m})$-approximate posted-price mechanism for $q$-partitioning valuations. \cref{section:introconcentration} overviews our main results on concentration inequalities. \cref{section:futurework} concludes.

The appendices contain all omitted proofs, along with some additional facts about the partitioning hierarchy. For example, \cref{section:closeness} discusses the distance of subadditive functions to $q$-partitioning functions.

\section{Preliminaries}\label{sec:prelim}
Throughout the entire paper, we assume that valuations
$f:2^M\longrightarrow \mathbb{R}^+$ are \textit{normalized}, meaning that $f(\emptyset) = 0,$ and \textit{increasing monotone}, meaning that $f(S)\le f(T)$ whenever $S\subseteq T.$\\

\noindent\textbf{Standard Valuation Classes.} A valuation function $f$ is \emph{subadditive} if for all $S,T$, $f(S \cup T) \leq f(S) + f(T)$. $f$ is \emph{XOS} if there exists a collection $\mathcal{A}$ of non-negative additive functions\footnote{A valuation function $v$ is non-negative additive if for all $S$, $v(S)=\sum_{ i \in S} v(\{i\}),$ where $v(\{i\})\geq 0$ holds for all $i.$} such that for all $S$, $f(S)=\max_{v\in \mathcal{A}}\{v(S)\}$. $f$ is \emph{fractionally subadditive} if for any $S$ and any fractional cover $\alpha(\cdot)$ such that for all $j \in S$ $\sum_{T\ni j} \alpha(T) \geq 1$ and $\alpha(T)\geq 0$ for all $T,$ it holds that $f(S) \leq \sum_{T} \alpha(T) f(T)$. It is well-known that $f$ is XOS if and only if it is fractionally subadditive via LP duality~\cite{Feige09}. \\

\noindent\textbf{PH-$k,$ MPH-\hmath$k,$ and Subadditive MPH-$k$ valuations}
Maximum over Positive Hypergraph-$k$ valuations, in short MPH-$k$, were introduced in \cite{FeigeFIILS15}. Since then, they have been studied in various different contexts such as communication complexity of combinatorial auctions~\cite{EzraFNTW19} and posted price mechanisms~\cite{FeldmanGL15}. One motivation behind MPH-$k$ valuations is to construct a hierarchy of valuation classes (starting with XOS) by replacing additive valuations with a richer class of valuations parameterized by $k$. Specifically:

\begin{definition}
A valuation $v:2^{[m]}\longrightarrow \mathbb{R}_{\ge 0}$ is:
\begin{enumerate}
    \item \textbf{PH-$k$} if there exist non-negative weights $w(E)$ for subsets $E\subseteq [m], |E| \leq k$, such that for all $S\subseteq [m]$: $\displaystyle v(S) = \sum_{T\subseteq S\; : \; |T|\le k} w(T).$
    \item \textbf{MPH-$k$} if there exists a set of PH-$k$ valuations $\mathcal{A}$ such that $\displaystyle v(S) = \max_{a\in \mathcal{A}} a(S)$ for all $S\subseteq[m].$
    \item \textbf{Subadditive MPH-$k$ (CFMPH-$k$)} if $v$ is simultaneously subadditive and MPH-$k.$ 
\end{enumerate}
\end{definition}

Note that PH-$1$ valuations are exactly the class of additive valuations, so the class of MPH-$1$ valuations is exactly the class of XOS valuations. Note also that PH-$2$ valuations need not be subadditive (and therefore, MPH-$2$ valuations need not be subadditive either). MPH-$m$ contains all monotone valuation functions, and all subadditive functions are MPH-$m/2$~\cite{EzraFNTW19}. We establish a connection between $q$-partitioning valuations and valuations that are MPH-$\lceil m/q\rceil$ and subadditive in \cref{lem:qpartandmphk}.




\section{The Partitioning Interpolation}\label{sec:defineqpart}
Here, we present our main definition. We give its more intuitive ``primal form'' as the main definition, and establish a ``dual form'' in \cref{sec:dualqpart}.

\begin{definition}
\label{def:qpartprimal}
Let $q \in [2,m]$ be an integer. A valuation $v:2^{[m]}\longrightarrow \mathbb{R}_{\ge 0}$ satisfies the $q$-partitioning property if for any $S\subseteq [m]$ and any partition $(S_1, S_2, \ldots, S_q)$ of $S$ into $q$ (possibly empty) disjoint parts, and any fractional covering $\alpha$ of $[q]$ (that is, any non-negative $\alpha(\cdot)$ such that for all $j \in [q], \sum_{T \ni j} \alpha(T) \geq 1$):

$$v(S) \leq \sum_{T \subseteq [q]} \alpha(T) \cdot v\left(\cup_{j \in T} S_j\right).$$
We refer to the class of $q$-partitioning valuations over $[m]$ as $\mathcal{Q}(q,[m])$.

\end{definition}

The intuition behind our definition is that $q$ captures the complexity of non-negative fractional covers under which the value of $v(\cdot)$ is non-diminishing. Subadditive valuations are only non-diminishing under very simple covers (covering $S\cup T$ by $S$ and $T$), while XOS valuations are non-diminishing under arbitrarily complex fractional covers. The parameter $q$ captures the desired complexity in between. We now establish a few basic properties of $q$-partitioning valuations.

We begin with the following nearly-trivial observations: First, for any fixed $q$ and $m$, the class $\classqm$ is closed under conic combinations.\footnote{That is, $\classqm$ is closed under linear combinations with non-negative coefficients.} This has 
implications for oblivious rounding of linear relaxations \cite{FeigeFT16}. Furthermore, for any fixed $q$ and $m$, the class $\classqm$ is closed under taking pointwise suprema, which means that one can use the ``lower envelope technique'' when approximating functions by $q$-partitioning functions \cite[Section 3.1]{FeigeFIILS15}. 

Now, we establish the three promised properties from ~\cref{sec:intro}. We begin by confirming that indeed the partitioning interpolation interpolates between fractionally subadditive and subadditive valuations.

\begin{proposition}
\label{prop:exisetnce}
For all $m$, the following relations between classes of $q$-Partitioning valuations hold:
$$
\text{XOS}([m]) = 
\mathcal{Q}(m, [m])\subsetneq 
\mathcal{Q}(m-1, [m])\subsetneq \cdots
\cdots \subsetneq
\mathcal{Q}(2, [m])
= \text{CF}([m]).
$$
\end{proposition}

We provide a complete proof of \cref{prop:exisetnce} in \cref{appendix:existenceproblem}. It is reasonably straight-forward to see that $\text{XOS}([m]) = 
\mathcal{Q}(m, [m])$, and that $\mathcal{Q}(2, [m])
= \text{CF}([m])$. It is also straightforward to see the inclusions in the chain (any partition with $q$ parts is also a partition with $q+1$ by adding an empty partition). We show that each inclusion is strict via the following proposition, whose complete proof appears in \cref{appendix:existenceproblem}.

\begin{proposition}
\label{prop:existenceproblem}
Consider a valuation $v$ over $[m]$ such that $v(S) = 1$ whenever $1 \le |S|\le m-1$ and $v(\emptyset) = 0$. The largest value $v([m])$ for which $v$ is $q$-partitioning is $\frac{q}{q-1}.$
\end{proposition}

We now show that $\mathcal{Q}(q,[m])$ is close to $\mathcal{Q}(q+1,[m])$ in a precise sense. Note that \cref{def:closeness} applied to $q=m$ is exactly the notion of closeness used in~\cite{BhawalkarR11}. 

\begin{definition}
\label{def:closeness}
Suppose that $0< \gamma\le 1.$
A class of valuations $\mathcal{G}$ over $[m]$ is $\gamma$-close to the class $\classqm$ if for any $g\in \mathcal{G},$ any $S\subseteq [m],$ any partition $(S_1, S_2, \ldots, S_q)$ of $S$ into $q$ parts, and any 
fractional cover $\alpha$ of $[q],$ it is the case that
$$
\sum_{{T}\subseteq [q]} \alpha({T})g(\bigcup_{i \in {T}}S_i)\ge \gamma g(S). 
$$
\end{definition}

We will see a further interpretation of \cref{def:closeness} in \cref{prop:gammacitycore}. For now, we simply present the following ``smoothness'' claim.

\begin{theorem}
\label{thm:smoothness}
$\mathcal{Q}(q+1,[m])$ is  $\frac{q-1}{q}$-close to $\mathcal{Q}(q,[m])$.
 \end{theorem}

 The proof of \cref{thm:smoothness} appears in \cref{section:properties}. Finally, we provide our first-principles definition of $q$-partitioning via a cost-sharing game. This aspect is more involved, so we overview the setup in \cref{section:costsharing}.

\subsection{Interpretation in Cost Sharing}
\label{section:costsharing}

\subsubsection{Recap: characterizing XOS via cost-sharing}
Consider a set $[m]$ of players who are interested in receiving some service. There is a 
cost for this service described by a monotone increasing normalized cost function $c: 2^{[m]}\longrightarrow \mathbb{R}.$ Here, $c(S)$ is the cost that players $S$ need to pay together so that each of them receives the service. A natural question to ask is: \textit{When is it the case that one can allocate the cost of the service between the community such that no subset of players $T\subseteq [m]$ is better off by forming a coalition and receiving the service on their own?} Formally, this question asks whether the $\core$ of the game is nonempty. The $\core$ is the set of all non-negative \textit{cost-allocation vectors} $\mathbf{p} = (p_1, p_2, \ldots, p_m)$ that satisfy $\sum_{i\in S}p_i \le c(S) \; \forall S\subseteq [m],$ and $\sum_{i\in [m]}p_i = c([m])$ \cite[Definition 15.3]{AGTbook}. We'll refer to the game parameterized by cost function $c(\cdot)$ restricted to players in $S$ as $\game(c,S)$.
This question is answered by the Bondareva-Shapley Theorem (see~\cite[Theorem 15.6]{AGTbook}). Applied to monotone normalized cost functions $c,$ the theorem states:

\begin{theorem}[\cite{Bondareva63,Shapley67}]
\label{thm:classicshapleycorethm}
The $\core$ $\game(c,[m])$ is non-empty if and only if for any non-negative fractional cover $\alpha$ of $[m]$ it is the case that $\sum_{S\subseteq[m]}\alpha(S)c(S) \ge c([m]).$
\end{theorem}

\noindent
An immediate generalization of this theorem, which appears in \cite[Section 1.1]{Feige09}, is:

\noindent
\begin{theorem}
\label{thm:xosinshapley}
The $\core$ of $\game(c,S)$ is non-empty for all $S$ if and only if $c$ is fractionally subadditive.
\end{theorem}

An interpretation of the above statement is the following. No matter what subset $S\subseteq [m]$ of players are interested in the service, we can always design a cost allocation vector (which vector can depend on $S$) such that all players in $S$ are better off by purchasing the service together rather than deviating and forming coalitions.
Since finding cores might be impossible (unless $c$ is fractionally subadditive), the following
relaxation of a $\core$ appears in coalitional game theory literature. A non-negative vector $\mathbf{p}$ is in the $\gammacore$ of the game if and only if it satisfies $\sum_{i\in S}p_i \le c(S) \; \forall S\subseteq [m],$ and $\gamma c([m])\le \sum_{i\in [m]}p_i \le c([m])$ \cite[Definition 15.7]{AGTbook}.
Again, one has equivalent statements to \cref{thm:xosinshapley,thm:classicshapleycorethm} using a $\gammacore.$ We only state the analogous statement for \cref{thm:xosinshapley}:

\begin{theorem}
\label{thm:classicshapleygammacore}
The $\gammacore$ of $\game(c,S)$ is non-empty for all $S$ if and only if $c$ is $\gamma$-close to XOS.
\end{theorem}

\subsubsection{$q$-partitioning via cost-sharing}
Consider instead a partition of players into $q$ (possibly empty) cities $S_1,\ldots, S_q$. We think of each city as a fully-cooperative entity that takes a single action.\footnote{Perhaps the city has an elected official that acts on behalf of the city's welfare, or perhaps the city's members have built enough trust that they can perfectly profit-share any gains the city gets.} The question of interest is whether $\citycore(S_1, S_2, \ldots, S_q)$ of the game is non-empty. $\citycore(S_1, S_2, \ldots, S_q)$ is the set of non-negative \textit{cost-allocation vectors} $\mathbf{p} = (p_1, p_2, \ldots, p_q)$ that satisfy $\sum_{i\in T}p_i \le c(\bigcup_{i\in T}S_i)$ for all $T\subseteq [q],$ and $\sum_{i\in [q]}p_i = c(\bigcup_{i \in [q]}S_i).$ Note that a vector in the $\citycore$ will incentivize cooperation as each subset of cities needs to pay at least as much if they choose to form a coalition. We parallel the theorems in the previous section with the following propositions. We'll refer to the above game as $\game(c,S,S_1,\ldots, S_q)$ when the normalized monotone cost function is $c$, players in $S$ are participating, and they are partitioned into cities $S_1,\ldots, S_q$. 

\begin{proposition}
\label{prop:maincostsharing}
The $\citycore$ of $\game(c,S,S_1,\ldots, S_q)$
is non-empty for all $S,S_1,\ldots, S_q$
if and only if $c$ is $q$-partitioning.
\end{proposition}

Again, the interpretation is simple. No matter
which people are interested in the service and how they are distributed between cities, we can design a cost allocation vector such that all cities are better off by purchasing the service together rather than forming coalitions. Finally, one can also relax the concept of a $\citycore$ to a $\gammacitycore$ as follows. This is the set of non-negative \textit{cost-allocation vectors} $\mathbf{p} = (p_1, p_2, \ldots, p_q)$ that satisfy $\sum_{i\in T}p_i \le c(\bigcup_{i\in T}S_i) \; \forall T\subseteq [q],$ and 
$\gamma c(\bigcup_{i \in [q]}S_i)\le \sum_{i\in [q]}p_i \le c(\bigcup_{i \in [q]}S_i).$ We can then also conclude:

\begin{proposition}
\label{prop:gammacitycore}
The $\gammacitycore$ of $\game(c,S,S_1,\ldots, S_q)$
is non-empty for all $S,S_1,\ldots, S_q$
if and only if $c$ is $\gamma$-close to $q$-partitioning.
\end{proposition}

\subsection{The Dual Definition and Relation to MPH Hierarchy}\label{sec:dualqpart}
Finally, we provide a dual view of the $q$-partitioning property (as in XOS vs.~fractionally subadditive), and relate $q$-partitioning to valuations that are MPH-$\lceil m/q \rceil$. First, we observe that the $q$-partitioning property can be reinterpreted as a claim about a linear program, opening the possibility of a dual definition.

\begin{observation}\label{obs:primal} A valuation function $f$ is $q$-partitioning if and only if for all $S$ and all partitions $(S_1,\ldots, S_q)$ of $S$ into $q$ (possibly empty) disjoint parts, the value of the following LP is $v(S)$:\footnote{Below, the variables are $\alpha(T)$ for all $T \subseteq [q]$.}
\begin{equation}
\label{eq:qpartprimalLP}
    \begin{split}
        \min &\sum_{{T}\subseteq [q]}v(\bigcup_{i \in T}S_i)\cdot \alpha(T), \text{ s.t.}\\
        & \sum_{T \ni j} \alpha(T)\ge 1 \; \; \; \forall j\in [q],\\
        &\alpha(T)\ge 0 \; \; \; \forall T\subseteq [q].
    \end{split}
\end{equation}
\end{observation}

The proof of \cref{obs:primal} is fairly immediate by observing that feasible solutions to the LP are exactly fractional covers, and that the objective function is exactly the bound on $v(S)$ implied by that fractional cover. We now state a ``dual'' definition of $q$-partitioning valuations. The equivalence with \cref{def:qpartprimal} is a simple application of linear programming, which we present in \cref{section:definitionequivalence}.

\begin{definition}
\label{def:qpartdual}
Let $2\le q \le m$ be integers. A valuation $v:2^{[m]}\longrightarrow \mathbb{R}_{\ge 0}$ satisfies the dual $q$-partitioning property if for any $S\subseteq [m]$ and any partition $(S_1, S_2, \ldots, S_q)$ of $S$ into $q$ disjoint parts, the following linear program has value at least $v(S):$
\begin{equation}
\tag{\text{Dual Definition}}
\label{eq:qpartdualLP}
    \begin{split}
        \max &\sum_{j \in [q]} p_j, \text{ s.t.}\\
        & \sum_{j \in T} p_j\le v(\bigcup_{j \in T} S_j) \; \; \; \forall T\subseteq [q],\\
        &p_j\ge 0 \; \; \; \forall j\in [q].
    \end{split}
\end{equation}
\end{definition}

\noindent
This dual definition allows us to establish the useful relationship between the partitioning and MPH hierarchies given in \cref{lem:qpartandmphk}.

\begin{proposition}
\label{lem:qpartandmphk}
A valuation over $[m]$ satisfying the $q$-partitioning property is MPH-$\lceil \frac{m}{q}\rceil$ and subadditive.\end{proposition}

\begin{proof}
To prove this statement, for each $S\subseteq [m],$ we will create a clause $w^S$ containing hyperedges of size at most $\lceil \frac{m}{q}\rceil,$ which takes value $v(S)$ at $S$ and for any $T\neq S,$ $w^S(T)\le v(S).$ This will be clearly enough as we can take the maximum over clauses $w^S.$
Take an arbitrary set $S$ and partition it into $q$ subsets $S_1,S_2, \ldots, S_q$ of almost equal size such that each subset has at most $\lceil \frac{m}{q}\rceil$ elements.  We will construct a clause of the form
$$
w^S = (S_1: p_1, \; S_2 : p_2,\;  \ldots, \; S_q:p_q ),
$$
where $p_i$ is the weight of set $S_i$ for each $i.$ Note that the weights $p_1, p_2, \ldots, p_q$ must satisfy
\begin{equation*}
    \begin{split}
        & \sum_i p_i = v(S),\\
        & \sum_{i \in I} p_i \le v(\bigcup_{i \in I}S_i)\; \forall
I\subseteq [q],\\
    & p_i\ge 0\; \forall i\in [q].
    \end{split}
\end{equation*}
The existence of such weights is guaranteed by \cref{def:qpartdual}, which completes the proof.\end{proof}

\begin{remark}
\normalfont
It should be noted that the converse statement does not hold true if $q\not \in \{2,m\}$. Let $k = \lceil \frac{m}{q}\rceil.$ The valuation $v(S):= \max\left( \binom{|S|}{k}, \frac{1}{2}\binom{m}{k}\right)$ over $[m]$ is MPH-$k$ and subadditive. However, it is simple to show that it is not $q$-partitioning. Split $[m]$ into $q$ sets of almost equal size $S_1, S_2, \ldots, S_q$ and consider the fractional cover $\alpha$ over $[q]$ assigning weight $\frac{1}{q-1}$ to all subsets of $[q]$ of size $q-1.$ A simple calculation shows that $v$ and $\alpha$ do not satisfy the $q$-partitioning property.
\end{remark}

\section{Main result I: Posted Price Mechanisms}\label{section:postedprices}
We consider the setup of \cite{FeldmanGL15}. Namely, there are $n$ buyers interested in a set of items $[m].$ The buyers' valuations come from a product distribution $\mathcal{D} = \mathcal{D}_1\times \mathcal{D}_2\cdots\times\mathcal{D}_n$, known to the seller. The optimal expected welfare is then $\textsc{OPT}(\mathcal{D}):=\mathbb{E}_{\vec{v} \leftarrow \mathcal{D}}[\max_{\text{Partitions } S_1,\ldots, S_n}\{\sum_{i=1}^n v_i(S_i)\}]$. The goal of the seller is to fix prices $p_1,\ldots, p_m$ so that the following procedure guarantees welfare at least $c\cdot \textsc{OPT}$ in expectation:
\begin{itemize}
\item Let $A$ denote the set of available items. Initially $A = [m]$.
\item Visit the buyers one at a time in adversarial order. When visiting buyer $i$, they will purchase the set $S_i:=\arg\max_{S \subseteq A}\{v_i(S) - \sum_{i \in S} p_i\}$, and update $A:=A \setminus S_i$. 
\end{itemize}

\begin{theorem}
\label{thm:postedpriceqpart}
    When all agents have $q$-partitioning valuations, there exists a
    $\Omega(\frac{\log \log q}{\log \log m})$-competitive
    posted price mechanism.\footnote{Unless explicitly indicated, logarithms have base 2 throughout the rest of this section.}
\end{theorem}

Note that this result matches asymptotically the best known competitive ratios for XOS (when $q = m,$ a constant ratio mechanism was proven in \cite{FeldmanGL15}) and CF valuations (when $q = 2,$ a $\Omega(\frac{1}{\log \log m})$-competitive posted price mechanism was proven in \cite{DuttingKL20}) and interpolates smoothly when $q$ is in between.

Like~\cite{DuttingKL20}, we first give a proof in the case when each $\mathcal{D}_i$ is a point-mass, as this captures the key ideas. A complete proof in the general case appears in \cref{sec:incompleteposted}.

Our proof will follow the same framework as~\cite{DuttingKL20}. To this end, let $p \in [0,1]$ be a real number. Denote by $\Delta(p)$ the set of distributions over $2^{[m]}$ such that $\prob_{S\leftarrow \lambda}[i \in S]\le p$ holds for all $\lambda \in \Delta(p)$ and all 
$i \in [m]$. The framework of~\cite{DuttingKL20} establishes the following:

\begin{lemma}[{\cite[Eq. (6)]{DuttingKL20}}]
\label{lem:minimaxgame}
A class of monotone valuations $\mathcal{G}$ over $[m]$ is given. If for any $v\in \mathcal{G},$ there exists a real number $p \in [0,1]$ (possibly depending on $v$), such that 
$$
\max_{\lambda\in \Delta(p)}
\min_{\mu \in \Delta(p)}\expect_{S\leftarrow \lambda, T\leftarrow \mu}[v(S\backslash T)]\ge 
\alpha \times v([m]),
$$
then there exists an $\alpha$-competitive posted price mechanism when all players have valuations in $\mathcal{G}.$
\end{lemma}

\cite{DuttingKL20} then show that when $\mathcal{G}$ is the class of all subadditive functions, such a $p$ exists for \linebreak $\alpha = \Theta(\frac{1}{\log \log m})$, but no better. In the rest of this section, we will show that when $\mathcal{G}$ is the set of $q$-partitioning valuations, the conditions of 
\cref{lem:minimaxgame} hold with $\alpha = \Omega \left(\frac{\log \log q}{\log \log m}\right)$. It is clear that (the deterministic case of) \cref{thm:postedpriceqpart} follows immediately from \cref{lem:minimaxgame} and \cref{prop:qpartprobexists}. It is worth noting that, while we leverage \cref{lem:minimaxgame} exactly as in~\cite{DuttingKL20}, the proof of \cref{prop:qpartprobexists} for general $q$ is quite novel in comparison to the $q=2$ (subadditive) case.

\begin{proposition}\label{prop:qpartprobexists} Let $v \in \mathcal{Q}(q,[m])$. Then there exists a real number $p\in [0,1]$ such that:

$$\max_{\lambda\in \Delta(p)}
\min_{\mu \in \Delta(p)}\expect_{S\leftarrow \lambda, T\leftarrow \mu}[v(S\backslash T)]\ge
\Omega \left(\frac{\log \log q}{\log \log m}\right)\times v([m]).$$
\end{proposition}

\begin{proof} Denote 
$$
g(p) = \max_{\lambda\in \Delta(p)}
\min_{\mu \in \Delta(p)}\expect_{S\leftarrow \lambda, T\leftarrow \mu}[v(S\backslash T)],
$$
$$
f(p) = 
\max_{\lambda\in \Delta(p)}
\expect_{S\leftarrow \lambda}[v(S)],
$$
and let $\lambda^p$ be a maximizing distribution in 
$\arg \max_{\lambda\in \Delta(p)}
\expect_{S\leftarrow \lambda}[v(S)].$\\

\noindent
Without loss of generality, assume that $q$ is a perfect power of 2, i.e. $q = 2^r$ (we can always decrease $q$ to a power of $2$ without changing the asymptotics of $\Omega(\frac{\log \log q}{\log \log m})$). 
Now, fix some $p\in \left(0,\frac{1}{16}\right].$ We will show that $
g(p)\ge \frac{1}{8}\left(f(p) -f(p^{\frac{r}{2}}) \right).$
The first step is the obvious bound 
$$
g(p)\ge \min_{\mu \in \Delta(p)}\expect_{S\leftarrow \lambda^p, T\leftarrow \mu}[v(S\backslash T)],
$$
which follows from the fact that we can choose $\lambda = \lambda^p.$
Now, we want to bound $\expect_{S\leftarrow \lambda^p, T\leftarrow \mu}[v(S\backslash T)].$ Fix a distribution $\mu.$ Let $S$ be drawn according to $\lambda^p$ and $T_1, T_2, \ldots, T_r$ be $r$ independent sets drawn according to $\mu.$ Then,
$$
\expect[v(S\backslash T)] = 
\expect\left[
\sum_{i = 1}^r\frac{1}{r}v(S\backslash T_i)
\right].
$$

\noindent
Now, we will use the $q$-partitioning property. Note that the sets $T_1, T_2 \ldots, T_r$ define a partitioning of $S$ into $2^r = q$  subsets. That is, for any $\vec{v}\in \{0,1\}^r,$ we can define 
\begin{equation}
\tag{Partitioning with $r$ sets}
\label{eq:partwithrsets}
S_{\vec{v}} = \{j\in S: \; j \in T_i \text{ for }v_i = 1 \text{ and }j \not \in T_i \text{ for }v_i =0\}. 
\end{equation}
According to this partitioning, define $A_0, A_1, A_2, \ldots, A_r$ as follows:
$$
A_t = \{j \in S\; : \; j \text{ belongs to exactly }t \text{ of the sets }T_i\}
 = \bigcup_{\vec{v}\; :\; 1^T\vec{v} = t}S_{\vec{v}}.
$$
Then, by the $q$-partitioning property, we know that 
$$
\frac{8}{r}\left(
v(S\backslash T_1) + v(S\backslash T_2) + \cdots + 
v(S\backslash T_r)
\right) + 
v(\bigcup_{j \ge \frac{7r}{8}}A_j)\ge v(S).
$$
Indeed, that is the case for the following reason. If $\vec{v}$ is such that $1^T\vec{v}< \frac{7r}{8},$ then $S_{\vec{v}}$ belongs to at least\linebreak $r - 1^T\vec{v}\ge \frac{r}{8}$ of the sets $S\backslash T_i,$ so it is ``fractionally covered'' by the term\linebreak
$\frac{8}{r}\left(
v(S\backslash T_1) + v(S\backslash T_2) + \cdots + 
v(S\backslash T_r)\right).$ 
If, on the other hand, $\vec{v}$ is such that $1^T\vec{v}\ge \frac{7r}{8},$ then it is fractionally covered by the term $v(\bigcup_{j \ge \frac{7r}{8}}A_j).$\\

\noindent
Now, let $A = \bigcup_{j \ge \frac{7r}{8}}A_j.$ We claim that each element $j \in [m]$ belongs to $A$ with probability at most $p^{r/2}$ (over the randomness in drawing $S \leftarrow \lambda^p$ and $T_1,\ldots, T_r \leftarrow \mu$, then defining $A$ as above). To prove this, we will use the classical Chernoff bound \cref{thm:chernoff} as follows. Let $Y_i$ be the indicator that $j \in T_i.$ Then, 
$\prob[Y_i = 1]\le p,$ so 
$
\expect[\sum_{i = 1}^rY_i] \le rp.
$ On the other hand, $j$ is in $A$ if and only if 
$\sum_{i = 1}^rY_i\ge \frac{7}{8}r.$ Now, let 
$\delta  = \frac{7}{8p}-1$ and let $\mu = rp.$ Then, by \cref{thm:chernoff},
$$
\prob\left[\sum_{i = 1}^rY_i\ge \frac{7}{8}r\right] \le
\prob\left[\sum_{i = 1}^rY_i\ge \mu (1 + \delta)\right] \le 
\left(
\frac{e^{\delta}}{(1+\delta)^{1+\delta}}
\right)^\mu\le 
\left(
\frac{e^{1+\delta}}{(1+\delta)^{1+\delta}}
\right)^{rp} =
\left(\frac{8ep}{7}\right)^{\frac{7r}{8}}\le p^{\frac{r}{2}},
$$
where the last inequality follows since $p\le\frac{1}{16}.$\\

\noindent
All of this together shows that 
$$
g(p) \ge 
\min_{\mu \in \Delta(p)}\expect_{S\leftarrow \lambda^p, T\leftarrow \mu}[v(S\backslash T)]\ge 
$$
$$
\frac{1}{r}\sum_{i  = 1}^r \expect[v(S\backslash T_i)] \ge 
\frac{1}{8}\expect[v(S)] - 
\frac{1}{8}\expect[v(A)] \ge 
\frac{1}{8}f(p) - \frac{1}{8}f(p^{\frac{r}{2}}),
$$
where the last inequality holds as each element appears in $A$ with probability at most $p^{r/2}.$ Using the same telescoping trick as in \cite{DuttingKL20}, we conclude as follows. Let $s = \lceil\log_{\frac{r}{2}}\log_{16} m^2\rceil.$ Then,
$$
\sum_{i = 0}^{s-1} g(16^{-(\frac{r}{2})^i})\ge 
\frac{1}{8}\sum_{i = 0}^{s-1}
f(16^{-(\frac{r}{2})^i})- 
f(16^{-(\frac{r}{2})^{i+1}}) \ge
\frac{1}{8}f\left(\frac{1}{16}\right) - \frac{1}{8}f\left(\frac{1}{m^2}\right).
$$
However, $f(\frac{1}{16})\ge \frac{1}{16}v([m])$ as shown by the distribution $\lambda$ which takes the entire set $[m]$ with probability $\frac{1}{16}$ and the empty set with probability $\frac{15}{16}.$ On the other hand, $f(\frac{1}{m^2})\le \frac{1}{m}v([m])$ since any distribution which takes each element with probability at most $\frac{1}{m^2}$ is non-empty with probability at most $\frac{1}{m}$ and, furthermore, $v$ is normalized monotone. All together, this shows that 
$$
\sum_{i = 0}^{s-1} g(16^{-(\frac{r}{2})^i})\ge 
\frac{1}{8}(\frac{1}{16} - \frac{1}{m})v([m]).
$$
For all large enough $m$ (say $m>32$), there exists exists some $p'$ such that 
$$
g(p') \ge \frac{1}{256\times s}v([m]) = 
\frac{1}{O(\log_r\log_2 m)}v([m]) = 
O\left(\frac{\log r}{\log \log m}\right)v([m]) = 
O\left(\frac{\log \log q}{\log \log m}\right)v([m]),
$$
which finishes the proof.
\end{proof}

\section{Main Result II: Concentration Inequalities}
\label{section:introconcentration}

In this section, we present our concentration inequalities for the partitioning interpolation. We begin by overviewing our results and their context in further detail, highlighting some proofs. We provide complete proofs in the subsequent section and 
\cref{section:concentration}.
.

\cite{Vondrak10} establishes that when $v$ is XOS, the random variable $v(S)$ has $\expect[v(S)]$-subgaussian lower tails. The proof follows by establishing that $1$-Lipschitz XOS functions of independent random variables are \emph{self bounding}, and applying a concentration inequality of~\cite{BoucheronLM00}. We begin by showing that $1$-Lipschitz $q$-partitioning functions are $(\lceil m/q\rceil,0)$-self bounding, and applying a concentration inequality of~\cite{McDiarmidR06,BoucheronLM09} to yield the following:

\begin{theorem}
\label{thm:selfboundingqpart}
Any 1- Lipschitz $q$-partitioning valuation $v$ over $[m]$ satisfies the following inequalities
$$
\prob\left[v(S)\ge \expect[v(S)] + t\right]\le
\exp\left(-\frac{1}{2}\times\frac{t^2}{\lceil\frac{m}{q}\rceil \expect[v(S)] + \frac{3\lceil \frac{m}{q}\rceil-1}{6}t}\right)\;\text { for }t\ge 0,
$$
$$
\prob[v(S)\le \expect[v(S)] - t]\le
\exp\left(-\frac{1}{2}\times\frac{t^2}{\lceil\frac{m}{q}\rceil \expect[v(S)]}\right)\;\text{ for }\expect[Z]\ge t\ge 0.
$$
\end{theorem}

We prove \cref{thm:selfboundingqpart} in \cref{section:concentration}. \cref{thm:selfboundingqpart} matches~\cite{Vondrak10} at $q=m$, and provides non-trivial tail bounds for $q = \omega(1)$. One should note, however, the the above inequality is useless when $q$ is constant, as it only implies $O(m\expect[v(S)])$-subgaussian behaviour, but it is well known that any 1-Lipschitz set function is $m$-subgaussian via McDiarmid's inequality. Our next inequality considers an alternate approach, based on state-of-the-art concentration inequalities for subadditve functions.\\

\cite{schectman} proves\footnote{Schechtman actually considers the non-normalized case and proves that $\prob[v(S)\ge (q=1)a + k]\le 
q^{-k}2^q,$ but reducing the factor from $q+1$ to $q$ in the normalized case is a trivial modification of the proof. It follows, in particular, from our \cref{thm:qparttailspecial}.} that whenever $v$ is normalized 1-Lipschitz subadditive, the following inequality holds for any real numbers $a>0, k>0,$ and integer $q\ge 2,$ \begin{equation}
\label{eq:schehtman}
\prob[v(S)\ge qa + k]\le 
q^{-k}2^q.
\end{equation} In particular, setting $a = \median[v(S)],q=2$, and integrating over $k=0$ to $\infty$ allows one to conclude the bound $\expect[v(S)]\le 2\median[v(S)] + O(1)$, which has proven useful, for example in \cite{RubinsteinW18}. While this inequality establishes a very rapid exponential decay, the decay only begins at  $2a$ as we need $q\ge 2$.\footnote{Recall that this is \emph{necessary}, and not just an artifact of the proof --- an example is given in~\cite{Vondrak10}.} What if we seek an upper tail bound of the form $\prob[v(S)\ge 1.1a + k]$ or, even more strongly, something of the form  $\prob[v(S)\ge (1 + {o_m(1)})a + k]$? Our next inequality accomplishes this:

\begin{theorem}
\label{thm:qparttailspecial}
Suppose that $v\in \mathcal{Q}(q,[m])$, and $S\subseteq [m]$ is a random set in which each element appears independently. Then the following inequality holds for any $ a \ge 0, k\ge 0,$ and integers $1\le r < s\le \log_2 q.$
    $$
    \prob\left[v(S)\ge \frac{r}{s}a+k\right]\le 
    \left(\frac{r}{s}\right)^{-k}2^\frac{r}{s}.
    $$
\end{theorem}

 The interesting extension for $q$-partitioning valuations via \cref{thm:qparttailspecial} is that one may take\linebreak $s = \log_2 q-1, r = \log_2 q$. From here, for example, one can again take $a$ to be the median of $v(S)$, and integrate from $k=0$ to $\infty$ to conclude that $\expect[v(S)]\le (1 + O(\frac{1}{\log q}))\median[v(S)] + O(\log q).$
In the very special case of $q=m$, we can also replace $\frac{r}{s}$ with any real $1+ \delta >0$ and obtain the $\expect[v(S)]\le \median[v(S)] + O(\sqrt{\median[v(S)]})$, which is the same extremely strong relationship implied by the $\expect[v(S)]$-subgaussian behaviour. We prove these simple corollaries of \cref{thm:qparttailspecial} in \cref{section:concentration} and now proceed to 
a proof of 
\cref{thm:qparttailspecial}. To do so, we need to make a detour and generalize Talagrand's work on the method of ``control by $q$-points''.

\subsection{A Probabilistic Detour: A New Isoperimetric Concentration Inequality}
\label{section:detourisoperimetry}
Suppose that we have a product probability space $\Omega = \prod_{i=1}^N \Omega_i$ with product probability measure $\prob.$ Throughout, in order to highlight our new techniques instead of dealing with issues of measurability, we will assume that the probability spaces are discrete and are equipped with the discrete sigma algebra. These conditions are not necessary and can be significantly relaxed (see~\cite[Section 2.1]{Talagrand01}).

For  $q+1$ points\footnote{In the current section, we will reserve the letter ``$q$'' to indicate the number of points $y^1, y^2,\ldots, y^{q}.$ While this choice might seem unnatural given that the current paper is all about ``$q$-Partitioning''  and ``$q$'' already has a different meaning, we have chosen to stick to Talagrand's notation. The reason is that the probabilistic approach we use is already known as ``control by $q$ points'' in literature \cite[Section 3]{Talagrand01}.} $y^1, y^2,\ldots, y^{q}, x$ in $\Omega,$ and an integer $1\le s \le q,$ we define
\begin{equation}\label{eq:definefs}
\begin{split}
    f^s(y^1, y^2,  & \ldots, y^q; x) :=\\
       & \Big{|}\Big{\{}i\in [N]: x_i \text{ appears less than $s$ times in the multiset }  \{y^1_i, y^2_i, \ldots, y^q_i\}\Big{\}}\Big{|}.
\end{split}
\end{equation}
       
\noindent
Using this definition, one can extend $f^s$ to subsets $A_1, A_2, A_3, \ldots, A_q$ of $\Omega$ as follows
$$
       f^s(A_1, A_2, \ldots, A_q; x) := 
       \inf\{f^s(y^1, y^2, \ldots, y^q; x) \; : \; y^i\in A_i\; \forall i\}.
$$
When $A = A_1, A_2, \ldots, A_q,$ the function $f^s(A_1, A_2, \ldots, A_q; x)$ intuitively defines a ``distance'' from $x$ to $A.$\footnote{This interpretation of $f^s$ as a ``distance'' is what motivates the name ``isoperimetric inequality'' (see \cite{Talagrand01}).} The definition of the function $f^s$ is motivated by and generalizes previous work of Talagrand \cite{Talagrand01,Talagrand96}. Our main technical result, the proof of which is deferred to 
\cref{section:concentration}, is the following. In it, $A_1, A_2, \ldots, A_q$ are fixed while $x$ is random, distributed according to $\prob,$ which is the aforementioned product distribution over $\Omega.$

\begin{theorem}
\label{thm:talagrandgenerals}
\label{thm:specialcaseqpartitioning}
Suppose that $\alpha\ge \frac{1}{s}$ is a real number and $t(\alpha, q,s)$ is the larger root of the equation $t + \alpha q t^{-\frac{1}{\alpha s}} = \alpha q + 1.$ Then, 
$$
\int_{\Omega}
t(\alpha, q, s)^{f^s(A_1, A_2, \ldots, A_q; x)}d\prob(x)\le 
\frac{1}{\prod_{i=1}^q \prob[A_i]^{\alpha}}.
$$
In particular, setting $A_1 = A_2 = \cdots = A_q = A, \alpha = \frac{1}{s}, t = \frac{q}{s},$ we obtain
$$\prob[f^s(\underbrace{A, A, \ldots, A}_{q}; x)\ge k]\le 
\left(\frac{q}{s}\right)^{-k}\prob[x \in A]^{-\frac{q}{s}}.$$
\end{theorem}

\noindent
Using this fact, we are ready to prove \cref{thm:qparttailspecial}.

\subsubsection{Proof of \cref{thm:qparttailspecial}}
\begin{proof}
Denote by $x$ the (random) characteristic vector of $S.$ The entries of $x$ are independent. Denote by $A=  \{y\in \{0,1\}^{[m]}\; : \; v(y)\le a\}.$ We claim that whenever $v(x) = v(S)\ge \frac{r}{s}a+k,$ it must be the case that $f^s(\underbrace{A,A,\ldots, A}_r; x)\ge k$. Indeed, suppose that this were not the case. Then, there must exist some $x\in \{0,1\}^{[m]}$ such that $v(x)\ge \frac{r}{s}a+k,$ but $f^s(\underbrace{A,A,\ldots, A}_r; x)< k.$ In particular, this means that there exist $r$ vectors $y^{1}, y^2, \ldots, y^r$ in $A$ such that 
$$
\Big{|}\Big{\{}i\in [m]: x_i \text{ appears less than $s$ times in the multiset }  \{y^1_i, y^2_i, \ldots, y^r_i\}\Big{\}}\Big{|}<k.
$$
Now, let $T$ be the set corresponding to the characteristic vector of $x$ and $T^i$ to $y^i.$ We also denote the following sets:
$$
M = \Big{\{}i\in [m]: x_i \text{ appears less than $s$ times in the multiset }  \{y^1_i, y^2_i, \ldots, y^r_i\}\Big{\}},
$$
$$
M_i = T\cap T_i\; \forall 1\le i \le r.$$
Now, observe that each element of $T\backslash M$ appears in at least $s$ of the $r$ sets  $M_1, M_2, \ldots, M_r.$ Furthermore, as $\log q\ge r$ and $v$ is $q$-partitioning, for the same reason as in \cref{eq:partwithrsets}, we know that 
$$
v(T\backslash M)\le \sum_{i=1}^r \frac{1}{s}v(M_i)\le 
\sum_{i=1}^r \frac{1}{s}v(T_i) = 
\sum_{i=1}^r \frac{1}{s}v(y^i)
\le 
\frac{ra}{s}.
$$
By the choice of $y^1, y^2, \ldots, y^r,$ we know that $|M|<k.$ As all marginal values of $v$ are in $[0,1],$ it follows that $v(M)<k.$ By subadditivity, we reach the contradiction
$$
\frac{ra}{s} + k \le 
v(y) = v(T)\le 
v(T\backslash M) + v(M)< 
\frac{ra}{s} + k.
$$
Therefore, whenever  $v(x) = v(S)\ge \frac{r}{s}a+k,$ it must be the case that $f^s(\underbrace{A,A,\ldots, A}_r; x)\ge k.$ The statement follows from \cref{thm:specialcaseqpartitioning}.
\end{proof}

\section{Conclusion}
\label{section:futurework}
We introduce the partitioning interpolation to interpolate between fractionally subadditive and subadditive valuations. We provide an interpretation of the definition via a cost-sharing game (as in~\cite{Bondareva63,Shapley67} for fractionally subadditive), and also show a relation to the subadditive MPH-$k$ hierarchy via a dual definition. We apply our definition in canonical domains (posted price mechanisms and concentration inequalities) where the fractionally subadditive and subadditive valuations are provably ``far apart'', and use the partitioning interpolation to interpolate between them. One technical nugget worth highlighting is Equation~\eqref{eq:partwithrsets}, which appears in the proofs of both \cref{prop:qpartprobexists} and \cref{thm:qparttailspecial} --- this idea may be valuable in future work involving the partitioning interpolation. We overview several possible directions for future work below.\\

\noindent\textbf{Precise closeness of $\mathcal{Q}(q,[m])$ to $\mathcal{Q}(q',[m])$.} We establish that $\mathcal{Q}(q,[m])$ is $\frac{q-1}{q}$-close to $\mathcal{Q}(q+1,[m])$. In \cref{remark:closenessgap}, we show $q$-partitioning valuations that are \emph{not} $\left(\frac{q^2-1}{q^2}+\epsilon\right)$-close to $\mathcal{Q}(q+1,[m])$ for any $\epsilon>0.$ It is interesting to understand what valuations in $\mathcal{Q}(q,[m])$ are furthest from $\mathcal{Q}(q+1,[m])$ (and how far they are). 

Similarly, it is interesting to understand for which $\beta,$ the class $\mathcal{Q}(q,[m])$ pointwise $\beta$-approximates $\mathcal{Q}(q+1,[m])$.\footnote{See~\cite{DevanurMSW15} for a formal definition of pointwise $\beta$-approximation.} Interestingly, a function is $\beta$-close to XOS if and only if it is pointwise $\beta$-approximated by XOS. But, it is not clear whether these two properties are identical for $q\neq m.$ Determining the precise relationship between the two properties is itself interesting. We note that one direction is easy. For any $q,$ being pointwise $\beta$-approximated by $\mathcal{Q}(q,[m])$ implies being $\beta$-close to $\mathcal{Q}(q,[m])$. So the open question is whether the converse is true.

We also note that~\cite{BhawalkarR11} resolves asymptotically the closeness of $\mathcal{Q}(2,[m])$ to $\mathcal{Q}(m,[m])$ at $\Theta(\log m)$. In \cref{section:closeness}, we show that simple modifications of their arguments imply that $\mathcal{Q}(2,[m])$ is $\Theta(\log_2 q)$-close to $\mathcal{Q}(q,[m])$ for any $q$.\\

\noindent\textbf{Constructing ``hard'' \hmath$q$-Partitioning valuations.}
Our two main results,
\cref{thm:postedpriceqpart,thm:qparttailspecial}, both establish the existence of desirable properties for $q$-partitioning valuations. However, to demonstrate tightness, one would need ``hard'' constructions of $q$-partitioning valuations that are not $(q+1)$-partitioning (recall that we have given constructions of valuation functions in $\mathcal{Q}(q,[m])\setminus \mathcal{Q}(q+1,[m])$, but they are not ``hard''). It does not appear at all straightforward to adapt constructions of ``hard'' valuations that are subadditive but not fractionally subadditive (e.g.~\cite{BhawalkarR11}) to create a valuation function in $\mathcal{Q}(q,[m])\setminus \mathcal{Q}(q+1,[m])$. Indeed, there is no previous construction of valuations that are subadditive MPH-$k$ but not subadditive MPH-$(k-1)$ (even without restricting to ``hard'' constructions). Constructing such functions (for which, e.g., the arguments made in \cref{section:postedprices,section:introconcentration} are tight) is therefore an important open direction.\\

\noindent\textbf{Applications in Algorithmic Game Theory.} Guarantees of posted-price mechanisms are perhaps the most notable domain within algorithmic game theory where fractionally subadditive and subadditive functions are far apart. Still, there are other settings where they are separated. For example, in best-response dynamics in combinatorial auctions, a dynamics leading to a constant fraction of the optimal welfare exists in the fractionally subadditive case, but there is a $O(\frac{\log \log m}{\log m})$ impossibility result in the subadditive case \cite{DuttingK22}. There are also constant-factor gaps between approximations achievable in polynomial communication complexity for combinatorial auctions~\cite{DobzinskiNS10,Feige09,EzraFNTW19}, and for the price of anarchy of simple auctions~\cite{RoughgardenST16}. Note that in the context of communication complexity of combinatorial auctions, \cref{lem:qpartandmphk} combined with the results for two players in 
\cite{EzraFNTW19} already imply improved communication protocols across the interpolation, but no lower bounds stronger than what can be inherited from fractionally subadditive valuations are known.\\

\noindent\textbf{Concentration Inequalities.} In the proof of \cref{thm:talagrandgenerals}, we have followed the approach of Talagrand in \cite{Talagrand96}. Dembo provides an alternative proof of the special case of this theorem for $s = 1$ using a more systematic approach 
for proving isoperimetry
based on
information inequalities in \cite{Dembo97}. It is interesting to understand to what extent the inequality in \cref{thm:talagrandgenerals} can be recovered using the information inequalities framework in \cite{Dembo97}. Moreover, the state-of-the-art provides two different approaches to concentration inequalities at the two extremes of the partitioning interpolation. This results in one approach yielding sharper guarantees near $q=m$, and the other near $q=2$. It is important to understand the (asymptotically) optimal tail bounds across the partitioning interpolation, and it is interesting to understand whether there is a unified approach that yields (asymptotically) optimal tail bounds across a broad range of $\{2,\ldots, m\}$.

\section*{Acknowledgements}
We are thankful to Jan Vondr{\'a}k for a helpful discussion which led to the inequalities based on Talagrand's method of control by $q$ points. We are also thankful to the anonymous reviewers for helpful remarks on the exposition. 

\printbibliography
\appendix

\section{Strict Inclusion of Classes}
\label{appendix:existenceproblem}
\begin{proof}[Proof of \cref{prop:existenceproblem}]
First, partition $[m]$ arbitrarily into $q$ nonempty parts  $S_1, S_2, \ldots, S_q.$ Then, by choosing the fractional cover $\alpha$ of $q$ for which $\alpha (I) = \frac{1}{q-1}$ whenever $|I| = q-1$ and $\alpha (I) = 0$ otherwise, we conclude that 
$$
v([m])\le \frac{1}{q-1}\sum_{i\in [q]}v(S\backslash S_i) = 1+ \frac{1}{q-1}.$$
We also need to prove that when $v([m]) = 1+ \frac{1}{q-1},$ the valuation $v$ actually satisfies the $q$-partitioning property. To do so, take any $S\subseteq [m],$ any partition of  $S$ into $q$ disjoint subsets $S_1, S_2, \ldots, S_q,$ and any fractional cover $\alpha$ of $[q].$ We need to show that the inequality in \cref{def:qpartprimal} is satisfied. First, note that unless $S = [m],$ the inequality is trivial. Furthermore, we can assume that all $q$ subsets $S_i$ are nonempty. Indeed, if exactly $r\le q$ of the sets are nonempty, then we essentially need to prove that $v$ satisfies the $r$-partitioning property. This follows by induction on $q$ as $\frac{1}{q-1}\le \frac{1}{r-1}.$ Finally, we can assume that $\alpha([q]) = 0.$ Otherwise, there are two cases. Either $\alpha([q]) \ge 1,$ which is trivial. Or,  we can modify $\alpha$ to $\alpha'$ by setting $\alpha'([q]) = 0$ and $\alpha'(T) = \frac{1}{1-\alpha([q])}\alpha(T)$ for $T\subsetneq [q].$ Proving the statement for $\alpha'$ directly implies if for $\alpha$ as well.
\\
Now, construct a fractional cover $\beta$ from $\alpha$ in the following way. For any $\emptyset\subsetneq I\subsetneq[q],$ define $c(I)$ to be an arbitrary set of size $q-1$ containing $I.$ Then, $\beta(J) = 0$ whenever $|J|\neq q-1$ and $\beta(J)  =\sum_{I \in c^{-1}(J)}\alpha(I)$ otherwise. Clearly, $\beta$ is also a fractional cover of $[q],$ and (using that $v(I) = 1$ whenever $\emptyset \subsetneq I \subsetneq [m]$), we compute
$$
\sum_{{I}\subseteq [q]} \alpha({I})v(\bigcup_{i \in {I}}S_i) = 
\sum_{\emptyset \subsetneq {I}\subsetneq [q]} \alpha({I}) = 
\sum_{J\subseteq [q] : \; |J| = q-1}\beta(J).
$$
Now, using that $\beta$ is a fractional cover of $[q],$ we conclude that 
$$
\sum_{J\subseteq [q] : \; |J| = q-1}\beta(J) = 
\frac{1}{q-1}\sum_{i \in [q]}\sum_{J :\; i \in J}\beta(J) \ge
\frac{1}{q-1}\sum_{i \in [q]}1 = \frac{q}{q-1} = v([m]),
$$
as desired.
\end{proof}

\section{Smoothness of The Partitioning Interpolation}
\label{section:properties}
Here, we prove the smoothness property of the $q$-partitioning interpolation given in \cref{thm:smoothness}. 

\begin{proof}[Proof of \cref{thm:smoothness}]
Take any $q$-partitioning valuation $v,$ subset $S\subseteq [m],$ partitioning $S_1, S_2, \ldots, S_{q+1}$ of $S,$ and fractional cover $\alpha$ of $[q+1].$ We want to show that 
$$
\sum_{I\subseteq [q+1]} \alpha(I)
v(\bigcup_{i\in I}S_i)\ge 
\frac{q-1}{q}v(S).
$$
We proceed as follows:
$$
(q+1)\sum_{I\subseteq [q+1]} \alpha(I)
v(\bigcup_{i\in I}S_i) = 
\sum_{J \subseteq [q+1] : |J| = q}
\sum_{I\subseteq [q+1]} \alpha(I)
v(\bigcup_{i\in I}S_i) \ge 
$$
$$
\sum_{J \subseteq [q+1]\; : \; |J| = q}\; \;
\sum_{T \subseteq J} v(\bigcup_{i \in T}S_i)\; \; 
\sum_{I \subseteq [q+1]\; : \;I\cap J = T} \alpha (I) \ge
\sum_{J \subseteq [q+1]\; : \;|J| = q}
v(\bigcup_{i \in J}S_i).
$$
The last inequality holds
since $\beta^J$ defined as $\beta^J(T):= \sum_{I \subseteq [q+1] : I\cap J = T} \alpha (I)$ is a fractional cover of $J,$ the valuation $v$ is $q$-partitioning, 
and $|J| = q.$\\

\noindent
Now, note that for any $\{i,j\}\subseteq [q+1],$ we have:
$$
\sum_{k\; : \; k \neq i, k \neq j}v(\bigcup_{t\; : \;t\neq k}S_t) + 
\frac{1}{2}( v(\bigcup_{t\; : \;t\neq i}S_t) + 
v(\bigcup_{t\; : \;t\neq j}S_t))\ge 
$$
$$
\sum_{k\; : \;k \neq i, k \neq j}v(\bigcup_{t\; : \;t\neq k}S_t)  + v(\bigcup_{t\; : \;t\neq i, t\neq j}S_t)\ge 
(q-1)v(S),
$$
as given by the partition $S_i\cup S_j$ and $S_k$ for $k\not \in \{i,j\}$ of $S$ into $q$ parts. Taking the sum over all pairs $i\neq j,$ we conclude that 
$$
\left(\binom{q}{2} + \frac{1}{2}q\right)\sum_{J \subseteq [q+1]\; : \;|J| = q}
v(\bigcup_{i \in J}S_i) \ge 
\binom{q+1}{2}(q-1)v(S).
$$
Therefore, 
$$
(q+1)\left(\binom{q}{2} + \frac{1}{2}q\right)\sum_{I\subseteq [q+1]} \alpha(I)
v(\bigcup_{i\in I}S_i)\ge 
\binom{q+1}{2}(q-1)v(S),
$$
from which the conclusion follows.
\end{proof}

\begin{remark}
\label{remark:closenessgap}
\normalfont
We can show that the ratio $\frac{q-1}{q}$ in \cref{thm:smoothness} cannot be improved beyond $ \frac{q^2 - 1}{q^2}$ using
the $q$-partitioning valuation constructed in \cref{prop:existenceproblem}. Namely, take $v$ so that $v(\emptyset) =0,
v(I) = 1$ whenever $0 <|I|<m$ and $v([m]) =\frac{q}{q-1}.$ Now, take any $q+1$ disjoint non-empty sets $S_1, S_2, \ldots, S_{q+1}$ and consider the fractional covering $\alpha$ of $[q+1]$
assigning weight $\frac{1}{q}$ to all $J\subseteq [q+1]$ of  size $q.$ Then, the fractional value 
$\displaystyle \sum_{J \subseteq [q+1]}\alpha(J)v(\cup_{j \in J}S_j)$
is $\frac{q+1}{q} = (1 - \frac{1}{q^2})\frac{q}{q-1}.$ Finding the optimal ratio in $[1 - \frac{1}{q}, 1 - \frac{1}{q^2}]$ is an interesting open problem.
\end{remark}

\section{Equivalence of the Dual and Primal Definitions}
\label{section:definitionequivalence}
The proof that \cref{def:qpartprimal} and \cref{def:qpartdual} are equivalent is a simple application of 
linear programming duality. Namely, in \cref{eq:qpartprimalLP}, for each constraint 
$$
\sum_{j \in I}
\alpha(I)\ge 1, 
$$
we create a non-negative dual variable $p_j \ge 0.$ Taking the dual program gives exactly 
\cref{eq:qpartdualLP}, so the two linear programs have the same optimal value (as both are trivially feasible).

\section{Chernoff Bounds}
\begin{theorem}[{\cite[Theorem 4.4]{Mitzenmacher05}}]
\label{thm:chernoff}
Suppose that $Y_1, Y_2, \ldots, Y_r$ are $r$ independent random variables taking values in $\{0,1\}.$ Let $\mu = \sum_{i=1}^r \expect[Y_i]$ be their mean and $\delta>0$ be an arbitrary real number. Then,
$$
\prob\left[\sum_{i = 1}^r Y_i \ge \mu(1 + \delta)\right]
\le 
\left(
\frac{e^\delta}{(1+\delta)^{1+\delta}}
\right)^{\mu}.
$$
\end{theorem}

\section{Incomplete Information Posted Price Mechanism}\label{sec:incompleteposted}
Here, we tackle the true Bayesian setting, when the posted price mechanism is not restricted to point-mass distributions. To do so, we first introduce some further notation from \cite{DuttingKL20}. Write $\boldv$ for the valuation vector coming from distribution $\mathcal{D} = \mathcal{D}_1\times \mathcal{D}_2\cdots\times\mathcal{D}_n.$ Respectively, write $\opt(\boldv)$ for (an) optimal allocation with valuations $\boldv,$ where $\opt(\boldv) =  (\opt_1(\boldv),\opt_2(\boldv),\ldots, \opt_n(\boldv)).$ Denote the social welfare under these valuations by $\boldv(\opt(\boldv)).$ Phrased in these terms, we need to design a posted price mechanism with expected social welfare $\Omega\left(\frac{\log \log q}{\log \log m}\right)\expect_\boldv[\boldv(\opt(\boldv))].$\\

\noindent
The set $\Delta$ is defined as in \cref{section:postedprices}. We also need the following extra notation, which extends $\Delta$ to $n$-tuples of distributions.
$$
\Gamma(p):=
\left\{\{\nu^i\}_{i=1}^n \in \Delta(p)| \sum_{i=1}^n 
\prob_{S\longleftarrow \nu^i}[j \in S]\le p \text{ holds for all }j \in [m]\right\}.
$$
Now, the equivalent statement to 
\cref{lem:minimaxgame} in the framework of 
\cite{DuttingKL20}
is the following.

\begin{lemma}[{\cite[p.268]{DuttingKL20}}]
\label{lem:minimaxgamebayes}
A class of monotone valuations $\mathcal{G}$ over $[m]$ and a product distribution 
$\mathcal{D} = \mathcal{D}_1\times \mathcal{D}_2\cdots\times\mathcal{D}_n$ over valuations in $\mathcal{G}$ are given.
If there exists a real number $p \in [0,1]$ (possibly depending on $\mathcal{D}$), such that 
$$
\expect_{\boldv}\left[
\sup_{\lambda \in \Gamma(p)}
\inf_{\mu \in \Delta(p)}
\sum_{i = 1}^n 
\expect_{S_i\longleftarrow \lambda^i, T_i \longleftarrow \mu}\boldv^i(S_i\backslash T_i)
\right]\ge 
\alpha\times \expect_\boldv[\boldv(\opt(\boldv))],
$$
then there exists an $\alpha$-competitive posted price mechanism for the product distribution $\mathcal{D}.$
\end{lemma}

\noindent
As in \cref{section:postedprices}, we prove that when $\mathcal{G}$ is the set of $q$-partitioning valuations, the above lemma holds with\linebreak $\alpha = \Omega\left(\frac{\log \log q}{\log \log m}\right).$ The incomplete information case, however, is more technically challenging, so we break the proof into several lemmas. First, we need some more notation. Denote
$$
g^\boldv(q) = 
\sup_{\lambda \in \Gamma(p)}
\inf_{\mu \in \Delta(p)}
\sum_{i = 1}^n 
\expect_{S_i\longleftarrow \lambda^i, T_i \longleftarrow \mu}[\boldv^i(S_i\backslash T_i)]
$$
for any valuation vector $\boldv$ and
$$
f^\boldv(q) = 
\sup_{\lambda \in \Gamma(p)}
\sum_{i = 1}^n 
\expect_{S_i\longleftarrow \lambda^i}\boldv^i(S_i).
$$
Again, we assume that $q>4$ is a perfect power of $2$ and $q = 2^r,s = \lceil\log_{\frac{r}{2}}\log_{16} m^2\rceil.$
Let\linebreak $\displaystyle L = \left\{ {- \left(\frac{r}{2}\right)^i}|
i \in \{0,1,\ldots, s-1\}\right\}
.$ Note that $|L| = s .$ Our first lemma is the following 

\begin{lemma}
\label{lem:postedpriceswithoutexpectation}
For any vector $\boldv$ of $q$-partitioning valuations, the following inequality holds
$$
\frac{1}{s}\sum_{\ell \in L}
g^\boldv(16^\ell)\ge 
\frac{1}{8s}
\left(\frac{1}{16} - \frac{1}{m}\right)
\boldv(\opt(\boldv)).
$$
\end{lemma}
\begin{proof}
Note that $f^\boldv(16^{-1})\ge \frac{1}{16}\boldv(\opt(\boldv)).$ Indeed, the vector of independent distributions $(\lambda^i)_{i=1}^n$ such that $S_i \longleftarrow \lambda^i$ satisfies $S_i = \opt_i(\boldv)$ with probability $1/16$ and $S_i = \emptyset$ with probability $15/16$ is in $\Gamma(16^{-1}).$ This is true because the sets $\opt_i(\boldv)$ are disjoint. Similarly, 
$f^\boldv(\frac{1}{m^2})\le \frac{1}{m}$ holds as shown in \cite[Lemma A.2]{DuttingKL20}. Therefore, using the telescoping trick in \cref{prop:qpartprobexists}, it is enough to show that 
\begin{equation}
\label{eq:tobetelescoped}
g^\boldv(16^\ell)\ge 
\frac{1}{8}\left(
f^\boldv(16^{\ell}) - 
f^\boldv(16^{\frac{\ell r}{2}})\right).
\end{equation}
holds for all $\ell \in L.$ In fact, we will prove something stronger. For any $p\le \frac{1}{16}$ and any $\lambda, \in \Gamma(p),\mu \in \Delta(p),$ there exists some $\sigma \in \Gamma(p^{\frac{r}{2}})$ such that 
\begin{equation}
\label{eq:onestepminimax}
\sum_{i = 1}^n 
\expect_{S_i\longleftarrow \lambda^i, T_i \longleftarrow \mu}[\boldv^i(S_i\backslash T_i)]\ge 
\frac{1}{8}\sum_{i = 1}^n 
\expect_{S_i\longleftarrow \lambda^i}[\boldv^i(S_i)] - 
\frac{1}{8}\sum_{i = 1}^n 
\expect_{A_i\longleftarrow \sigma^i}[\boldv^i(A_i)].
\end{equation}
Taking suprema on both sides yields \cref{eq:tobetelescoped}. We prove this fact in similar manner to \cref{prop:qpartprobexists}. Let $T_{i,u}$ for  $1\le u \le r$ be $rn$ independent sets from the distribution $\mu.$ Then, for any $i \in [n],$
$$
\expect_{S_i\longleftarrow \lambda^i, T_i \longleftarrow \mu}[\boldv^i(S_i\backslash T_i)] = 
\expect_{S_i\longleftarrow \lambda^i, T_{i,u} \longleftarrow \mu}\left[\frac{1}{r}\sum_{u = 1}^r
\boldv^i(S_i\backslash T_{i,u})\right].
$$
Now, let $A_i$ be the set of elements of $S_i$ that appear in more than $\frac{7r}{8}$ of the sets 
$T_{i,1}, T_{i,2}, \ldots, T_{i,r}.$ Let $\sigma^i$ be the distribution of $A_i.$ The same reasoning as in \cref{prop:qpartprobexists} shows that 
$$
\expect_{S_i\longleftarrow \lambda^i, T_{i,u} \longleftarrow \mu}\left[\frac{1}{r}\sum_{u = 1}^r
\boldv^i(S_i\backslash T_{i,u})\right]\ge 
\frac{1}{8}\expect_{S_i \longleftarrow \lambda^i}[\boldv^i(S_i)] - 
\frac{1}{8}\expect_{A_i \longleftarrow \sigma^i}[\boldv^i(A_i)].
$$
Thus, summing over $i,$ we conclude that 
$$
\sum_{i = 1}^n 
\expect_{S_i\longleftarrow \lambda^i, T_i \longleftarrow \mu}[\boldv^i(S_i\backslash T_i)]\ge 
\frac{1}{8}\sum_{i = 1}^n \expect_{S_i \longleftarrow \lambda^i}[\boldv^i(S_i)] - 
\frac{1}{8}\sum_{i = 1}^n \expect_{A_i \longleftarrow \sigma^i}[\boldv^i(A_i)].
$$
To conclude, we simply need to show that $\sigma \in \Gamma(p^{\frac{r}{2}}).$ This holds for the following reason. Take some $j \in [m].$ Define $p_i = \prob_{S_i\longleftarrow \lambda^i}[j \in S_i].$ By the definition of $\Gamma(p),$ we know that $\sum_{i=1}^n p_i \le p.$ On the other hand, as each $p_i$ satisfies $p_i \le p \le \frac{1}{16},$ as in \cref{prop:qpartprobexists}, we conclude that 
$$
\prob_{A_i \longleftarrow \sigma^i}
[
j \in A_i
]\le 
p_i^{r/2}.
$$
Thus, all we need to show is that 
$$
\sum_{i = 1}^n
p_i^{r/2}\le p^{r/2}.
$$
This, however, is simple. The map $x\longrightarrow x^{r/2}$ is convex. Thus, its maximum on the simplex defined by $0\le p_i\; \forall i \in [n], \sum_{i = 1}^n p_i\le p$ is attained at a vertex. This exactly corresponds to $$
\sum_{i = 1}^n
p_i^{r/2}\le p^{r/2}.
$$
\end{proof}

\noindent
The statement we want is a simple corollary of \cref{lem:postedpriceswithoutexpectation}. 

\begin{lemma}
\label{lem:bayesianpostedprice}
If $\mathcal{D}$ is a product distribution over $q$-partitioning valuations, there exists some deterministic $p\in [0,1]$ (potentially 
depending on $\mathcal{D}$
) such that 
$$
\expect_\boldv\left[
g^\boldv(p)
\right] \ge 
\Omega\left(\frac{\log \log q}{\log \log m}\right)
\expect_\boldv[\boldv(\opt(\boldv))].
$$
\end{lemma}
\begin{proof}
Using the uniform distribution $Unif(L)$ over $L,$ we conclude from \cref{lem:postedpriceswithoutexpectation}
that $$
\expect_{\ell \longleftarrow Unif(L)}[
g^\boldv(16^\ell)] \ge \frac{1}{8s}
\left(\frac{1}{16} - \frac{1}{m}\right)
\boldv(\opt(\boldv)).
$$
Taking the expectation over $\boldv$ and interchanging order of expectations, we obtain
$$
\expect_{\ell \longleftarrow Unif(L)}
\left[
\expect_\boldv[
g^\boldv(16^\ell)]
\right]
\ge \frac{1}{8s}
\left(\frac{1}{16} - \frac{1}{m}\right)
\expect_\boldv[\boldv(\opt(\boldv))].
$$
Therefore, as $g^\boldv$ is clearly non-negative, 
the moment method implies the existence of some $\ell_1$ such that 
$$
\expect_\boldv[
g^\boldv(16^{\ell_1})] \ge \frac{1}{8s}
\left(\frac{1}{16} - \frac{1}{m}\right)
\expect_\boldv[\boldv(\opt(\boldv))].
$$
Using that $\frac{1}{s} = \Omega\left(\frac{\log \log q}{\log \log m}\right),$ we complete the proof.
\end{proof}

\noindent
The statement of 
\cref{thm:postedpriceqpart} for the incomplete information case follows from
\cref{lem:minimaxgamebayes} and\linebreak 
\cref{lem:bayesianpostedprice}.

\section{Concentration Inequalities}
\label{section:concentration}
We discuss in full detail the concentration inequalities in \cref{section:introconcentration} and their more general versions.  The setup that we consider throughout the rest of this section is the following.
A valuation $v$ over $[m]$ is given.
$S\subseteq [m]$ is a random set such that each item $j\in [m]$ appears independently in $S$ (different items might appear with different probabilities). We study how concentrated $v(S)$ is around a point of interest such as its mean or median.  
Before delving into the main content, we make one important note. Concentration inequalities and tail bounds depend on the scale of the valuations. That is, the valuation $2v$ has a ``weaker'' concentration than the valuation $v$ (at least when considering an additive deviation from the mean of the form $\expect[v(S)]\pm t$). For that reason, throughout we assume that all marginal values, i.e. values $v(S\cup\{i\}) - v(S)$ for $S\subseteq [m], i\in [m],$ are between 0 and 1. Note that if $v$ is subadditive, in particular $q$-partitioning for some $q,$ this is equivalent to $v(\{i\})\le 1\; \forall i\in [m]$.\\

\noindent
Throughout the rest of the section, by abuse of notation, we will write $v(S)$ and $v(x_1, x_2, \ldots, x_m)$ interchangeably where $(x_1, x_2, \ldots, x_m)$ is the characteristic vector of $S$ in $\{0,1\}^m.$

\subsection{Concentration via Self-Bounding Functions}
\label{section:selfbounding}
Vondrak shows that XOS and submodular functions exhibit strong concentration via self-bounding functions \cite[Corollary 3.2]{Vondrak10}. We adopt this approach and show how the bounds generalize to MPH-$k$ valuations with bounded marginal values and, in particular, to $q$-partitioning valuations (naturally, the concentration becomes weaker as $q$ decreases). This approach relies on the method of self-bounding functions.

\begin{definition}[{
\cite[Definition 2.3]{Vondrak10}}]
\label{def:selfboundingfunctions}
A function $f:\{0,1\}^m \longrightarrow \mathbb{R}$ is $(a,b)$-self-bounding if there exist real numbers $a,b>0$ and functions $f_i:\{0,1\}^{m-1}\longrightarrow \mathbb{R}$ such that if we denote\linebreak $x^{(i)} = (x_1, x_2, \ldots, x_{i-1}, x_{i+1}, \ldots, x_m),$ then for all $x\in \{0,1\}^m$ and $i\in [m]:$
$$
0\le f(x)-f_i(x^{(i)})\le 1, \text{ and }$$
$$
\sum_{i=1}^m (f(x) - f_i(x^{(i)}))\le 
af(x) + b. 
$$
\end{definition}

\noindent
When $f$ is monotone, the optimal choice of $f_i$ is clearly given by\linebreak
$f_i(x^{(i)}):=f(x_1, x_2, \ldots, x_{i-1},0, x_{i+1}, \ldots, x_m).$ We adopt this approach in \cref{lem:MPHslefbounding}.
An application of the entropy method gives the following concentration for self-bounding functions.

\begin{theorem}[{\cite[Theorem 3.3]{Vondrak10}}]If $f:\{0,1\}^m\longrightarrow \mathbb{R}$ is $(a,b)$-self-bounding for some $a\ge \frac{1}{3},$ and $X_1, X_2, \ldots, X_m$ are independent, then for $Z = f(X_1, X_2, \ldots, X_m)$ and $c = \frac{3a-1}{6}:$
\begin{equation*}
\begin{split}
    & \prob[Z\ge \expect[Z] + t]\le
\exp\left(-\frac{1}{2}\times\frac{t^2}{a\expect[Z] + b + ct}\right)\; \text{ for any }0<t,\\
    &
     \prob[Z\le \expect[Z] - t]\le
     \exp\left(-\frac{1}{2}\times\frac{t^2}{a\expect[Z] + b}\right)\;\text{ for any }0<t<\expect[Z].
\end{split}
\end{equation*}
\end{theorem}

\noindent
Vondrak \cite[Lemma 2.2]{Vondrak10} shows that XOS valuations (also $m$-partitioning or MPH-1) are $(1,0)$-self-bounding.
We generalize as follows.

\begin{proposition}
\label{lem:MPHslefbounding}
Any (not necessarily subadditive) MPH-$k$ valuation with all marginal values in $[0,1]$ is $(k,0)$-self-bounding. In particular, this holds for all subadditive MPH-$k$ valuations $f$ such that\linebreak $f(\{i\})\le 1\; \forall i \in [m].$
\end{proposition}
\begin{proof}
We make the canonical choice $f_i(x^{(i)}):=(x_1, x_2, \ldots, x_{i-1},0, x_{i+1}, \ldots, x_n)$ in the proof. In other words, $f_i(S):=f(S\backslash \{i\}).$ 
Note that $0\le f(x)-f_i(x^{(i)})\le 1$ follows from the fact that $f$ is monotone and all marginal values are bounded by 1.
We need to prove that for any $S\subseteq [m],$ it is the case that 
$$
\sum_{j \in [m]}(f(S) - f(S\backslash\{j\}))\le 
kf(S)$$
The latter is clearly equivalent to 
$$
\sum_{j \in S}(f(S) - f(S\backslash\{j\}))\le 
kf(S),
$$
which is the statement of \cite[Lemma 6.2]{EzraFNTW19}.
\end{proof}

\noindent
Combining the above propositions with \cref{lem:qpartandmphk}, we immediately deduce \cref{thm:selfboundingqpart}.\\

\noindent
Note that this result is potentially useful only when $\expect[v(S)] \ll q.$ Otherwise Azuma's inequality \cite{Azuma} implies that $v(S)$ is $m$-subgaussian. In fact, this holds for any 1-Lipschitz (not necessarily subadditive!) function $f:\{0,1\}^m\longrightarrow \mathbb{R}.$\\

\noindent
So, can we derive any tail bounds for
small values of $q$ (in particular, subadditive valuations) which are better than what is already true for any $1$-Lipschitz valuations over $[m]$? It turns out that the answer to this question is ``yes'' and this is the subject of the next section. Before turning to it, however, we note that we cannot hope to do much better than \cref{thm:selfboundingqpart}
using the method of self-bounding functions.

\begin{example}
\normalfont
For any $2\le q\le m,$ there exists a symmetric $q$-partitioning valuation over $[m]$ with marginal values in $[0,1]$ which is not $(a,0)$-self-bounding for any $a<\frac{m}{q}.$ Indeed, consider the $q$-partitioning function $v$ constructed in \cref{prop:existenceproblem}, defined by $v(\emptyset) = 0, v(S) = 1$ for $\emptyset \subsetneq S\subsetneq [m],$ and 
$v([m]) = \frac{q}{q-1}.$ Suppose, for the sake of contradiction, that $v$ is $(a,0)$-self-bounding for some $a<\frac{m}{q}$ and $v_i.$
Then, for any $S$ with characteristic vector $x$ and $i\in [m],$ it must be the case that $v_i(x^{(i)})\le f(S\backslash \{i\}) \le 1$ by the first inequality in \cref{def:selfboundingfunctions}. As a result, when we set $x$ to be the all-ones vector in the second inequality, we conclude that $a\frac{q}{q-1}\ge \frac{m}{q-1}.$ Thus, $a\ge \frac{m}{q},$ which is a contradiction.
\end{example}

\subsection{Tail Bounds via an Isoperimetric Inequality}
We now 
take a deeper look into the stated isoperimetric inequalities in \cref{section:introconcentration}. To do so, we first present some background on
\cref{thm:talagrandgenerals} and the proof of the statement. 
\subsubsection{The Isoperimetric Inequality}
We consider the setup introduced in \cref{section:detourisoperimetry}.
In \cite[Section 3.1.1]{Talagrand01} the author considers the special case $s = 1,$ and in \cite[Section 5.7]{Talagrand96} he considers the special case $s = q-1.$ In particular, the inequalities he proves are the following.

\begin{theorem}[{\cite[Section 3.2]{Talagrand01}}]
\label{thm:talagrands1} Suppose that $\alpha>0$ is a real number and $z(q,\alpha)$ is the larger root of the equation $z+q\alpha z^{-\frac{1}{\alpha}} = 1+q\alpha.$
Then
$$
\int_{\Omega}
z(q,\alpha)^{f^1(A_1, A_2, \ldots, A_q; x)}d\prob(x)\le 
\frac{1}{\prod_{i=1}^q \prob[A_i]^\alpha}.
$$
In particular, setting $A= A_1 = A_2= \ldots = A_q, \alpha = 1, z= q,$ one has
$$\prob[f^1(\underbrace{A, A, \ldots, A}_{q}; x)\ge k]\le 
q^{-k}\prob[A]^{-q}.$$
\end{theorem}

\noindent
\cref{thm:talagrands1} is the main tool in Schechtman's bound (see \cref{section:introconcentration}).
The result for $s = q-1$ looks similar:

\begin{theorem}
\label{thm:talagrandsq1}
(\cite[Theorem 5.4]{Talagrand96}) Suppose that $\tau$ is the positive root of $e^{\tau/2} +e^{-\tau} = 2.$ Then,
$$
\int_{\Omega}
e^{\frac{\tau}{q}f^{q-1}(A_1, A_2, \ldots, A_q; x)}d\prob(x)\le 
\frac{1}{\prod_{i=1}^q \prob[A_i]^{1/q}}.
$$
In particular, setting $A= A_1 = A_2= \ldots = A_q$ one has
$$\prob[f^{q-1}(\underbrace{A, A, \ldots, A}_{q}; x)\ge k]\le 
e^{-\frac{\tau k}{q}}\prob[A]^{-1}.$$
\end{theorem}

\noindent
In this section, we provide a uniform view on 
\cref{thm:talagrands1} and \cref{thm:talagrandsq1} by proving \cref{thm:talagrandgenerals}.
First, we will demystify the number $t(\alpha, q, s).$ It comes from the following fact.

\begin{lemma}
\label{lemma:talphaqs}
Suppose that $x_1, x_2, \ldots, x_q$ are real numbers in $[0,1]$ and $\alpha \ge \frac{1}{s}.$ Then, for $t(\alpha,q,s)$ defined as in \cref{thm:talagrandgenerals}, one has
$$
\min \left(
t(\alpha, q, s),
\min_{1\le i_1<i_2<\cdots <i_s\le q}
(x_{i_1}x_{i_2}\cdots x_{i_s})^{-\alpha}
\right)
+ \alpha \sum_{i = 1}^q x_i \le 
\alpha q + 1.
$$
Furthermore, $t(\alpha,q,s)$ is the largest $t$ with this property. We use the convention that $0^{-\alpha}  = +\infty.$
\end{lemma}

\noindent
As the proof of \cref{lemma:talphaqs} provides no insight and is just a computation, we postpone it to \cref{section:appendixtalphaqs}.
We also remark that a version of \cref{thm:talagrandgenerals} holds for $\alpha <\frac{1}{s}$ as well, but the definition of $t(\alpha,q,s)$ becomes even more complicated (see \cref{rmk:alphalessthansiso}).
We now turn to the proof of \cref{thm:talagrandgenerals}.

\begin{proof}[Proof of \cref{thm:talagrandgenerals}] We follow closely the proof of \cref{thm:talagrandsq1} due to Talagrand in \cite{Talagrand96}. We proceed by induction over $N,$ the dimension of the product space.

\noindent
\textit{Base:}
When $N = 1,$ let $g_i$ be the indicator function of $A_i$ for $1\le i \le q.$ Observe that 
$$
\int_{\Omega}
t(\alpha, q, s)^{f^s(A_1, A_2, \ldots, A_q; x)}d\prob(x) = 
$$
$$
\int_{\Omega}
\min \left(
t(\alpha, q, s),
\min_{1\le i_1<i_2<\cdots <i_s\le q}
(g_{i_1}(x)g_{i_2}(x)\cdots g_{i_s}(x))^{-\alpha}
\right)
d\prob(x).
$$
Indeed, this is true because $f^s(A_1, A_2, \ldots, A_q; x) = 0$ when there exist $s$ sets $A_{i_1}, A_{i_2}, \ldots, A_{i_s}$ to which $x$ belongs and $f^s(A_1, A_2, \ldots, A_q; x) = 1$ otherwise simply by the definition $f^s.$ Using \cref{lemma:talphaqs},
we conclude that 
$$
\int_{\Omega}
\min \left(
t(\alpha, q, s),
\min_{1\le i_1<i_2<\cdots <i_s\le q}
(g_{i_1}(x)g_{i_2}(x)\cdots g_{i_s}(x))^{-\alpha}
\right)
d\prob(x)\le 
$$
$$
1 + \alpha q - \alpha \sum_{i=1}^q\int_{\Omega} g_i(x)d\prob(x) = 
1 + \alpha q - \alpha\sum_{i=1}^q\prob[A_i].
$$
Now, we will use twice the well known inequality $1 + \log x \le x$ as follows
$$
1 + \alpha q - \alpha\sum_{i=1}^q\prob[A_i] = 
1 + \alpha \sum_{i=1}^q (1 - \prob[A_i]) \le 
1 + \alpha \sum_{i = 1}^q \log \frac{1}{\prob[A_i]} = 
$$
$$
1 + \log \prod_{i=1}^q \frac{1}{\prob[A_i]^\alpha} \le
\prod_{i=1}^q \frac{1}{\prob[A_i]^\alpha}, 
$$
with which the base case is completed.\\

\noindent
\textit{Inductive Step:} Now, let $A_1, A_2, \ldots, A_q$ all belong to $\Omega = \Omega'\times \Omega_{N+1},$ where $\Omega' = \prod_{i=1}^N\Omega_i.$ For each $i\in [q],w\in \Omega_{N+1}$ define the following sets:
$$
A_i(w) = \{x\in \Omega'\;  : \; (x,w) \in A_i \},
$$
$$
B_i = \bigcup_{w\in \Omega_{N+1}}A_i(w).
$$
Fix some $w\in \Omega_{N+1}.$ For $I\subseteq [q]$ with $|I| = s,$ denote 
$C^I_i = A_i(w)$ whenever $i \in I$ and $C^I_i = B_i$ whenever $i \not \in I.$ Then, we can make the following observations:
$$
f^s(A_1, A_2, \ldots, A_q; (x,w))\le 
1 + f^s(B_1, B_2, \ldots, B_q; x)\; \forall (x,w)\in \Omega,
$$
$$
f^s(A_1, A_2, \ldots, A_q; (x,w))\le 
 f^s(C^I_1, C^I_2, \ldots, C^I_q; x)\; \forall (x,w)\in \Omega, I\subseteq [q] \text{ with }|I| = s.
$$
Indeed, the first inequality follows from the following fact. If $(b^1, b^2, \ldots, b^q)\in B_1\times B_2\times\cdots \times B_q,$ then for each $i, $ we can find some $a^i = (b^i, w^i)\in A_i.$ Clearly, 
$$
f^s(a^1, a^2, \ldots, a^q; (x,w))\le 
1 + f^s(b^1, b^2, \ldots, b^q; x)
$$
as the only extra coordinate that may appear less than $s$ times in the required multiset in \cref{eq:definefs} is the $N+1$'th coordinate $w.$ The second inequality follows from the same fact except that we choose $a^i = (b^i,w)$ whenever $i\in I.$\\

\noindent
Having those two inequalities, we can fix $w$ and compute:
$$
\int_{\Omega'}
t(\alpha, q, s)^{f^s(A_1, A_2, \ldots, A_q; (x,w))}d\prob(x)\le 
$$
$$
\int_{\Omega'}
\min\left(
t(\alpha, q, s)^{1+f^s(B_1, B_2, \ldots, B_q; x)},
\min_{_{|I|  =s}}
t(\alpha, q, s)^{
f^s(C^I_1, C^I_2, \ldots, C^I_q; x)
}
\right)\le 
$$
$$
\min\left(
\int_{\Omega'}
t(\alpha, q, s)^{1+f^s(B_1, B_2, \ldots, B_q;x)},
\min_{_{|I|  =s}}
\int_{\Omega'}
t(\alpha, q, s)^{
f^s(C^I_1, C^I_2, \ldots, C^I_q; x)
}
\right).
$$
Now, we just use the inductive hypothesis as each $B_i$ and $C^I_i$ is in the space $\Omega',$ which is a product of $N$ spaces. We bound the above expression by
$$
\min\left(
t(\alpha, q,s)
\frac{1}{\prod_{i=1}^q \prob[B_i]^\alpha}, 
\min_{|I| = s}
\frac{1}{\prod_{i=1}^q \prob[C^I_i]^\alpha}
\right) = 
$$
$$
\frac{1}{\prod_{i=1}^q \prob[B_i]^\alpha}
\min\left(
t(\alpha, q,s), 
\min_{1\le i_1<i_2<\ldots <i_s\le q}
\frac{\prob[B_{i_1}]^\alpha}{\prob[A_{i_1}(w)]^\alpha}
\frac{\prob[B_{i_2}]^\alpha}{\prob[A_{i_2}(w)]^\alpha}
\cdots
\frac{\prob[B_{i_s}]^\alpha}{\prob[A_{i_s}(w)]^\alpha}
\right)\le
$$
$$
\frac{1}{\prod_{i=1}^q \prob[B_i]^\alpha}
\left(
1 + \alpha q - \alpha\sum_{i=1}^q\frac{\prob[A_i(w)]}{\prob[B_i]}\right),
$$
where in the last inequality we used \cref{lemma:talphaqs} since $\frac{\prob[A_i(w)]}{\prob[B_i]}\le 1$ for each $i.$ However, as we are dealing with a product measure, by Tonelli's theorem for non-negative functions, it follows that  
$$
\int_{\Omega}
t(\alpha, q, s)^{f^s(A_1, A_2, \ldots, A_q; (x,w))}d\prob(x,w)= 
$$
$$
\int_{\Omega_{n+1}}\int_{\Omega'}
t(\alpha, q, s)^{f^s(A_1, A_2, \ldots, A_q; (x,w))}d\prob(x)d\prob(w)\le 
$$
$$
\int_{\Omega_{n+1}}
\frac{1}{\prod_{i=1}^q \prob[B_i]^\alpha}
\left(
1 + \alpha q - \alpha\sum_{i=1}^q\frac{\prob[A_i(w)]}{\prob[B_i]}
\right)
d\prob(w) = 
$$
$$
\frac{1}{\prod_{i=1}^q \prob[B_i]^\alpha}
\int_{\Omega_{n+1}}
\left(
1 + \alpha q - \alpha\sum_{i=1}^q\frac{\prob[A_i(w)]}{\prob[B_i]}
\right)
d\prob(w).
$$
Using the same approach as in the base case, but this time for the functions $g_i(w):=\frac{\prob[A_i(w)])}{\prob[B_i]},$ we bound the last expression by 
$$
\frac{1}{\prod_{i=1}^q \prob[B_i]^\alpha}\times 
\frac{1}{\prod_{i=1}^q \left(\frac{\prob[A_i]}{\prob[B_i]}\right)^\alpha} = 
\frac{1}{\prod_{i=1}^q \prob[A_i]^\alpha},
$$
as desired.
\end{proof}

\subsubsection{Tail Bounds and Median-Mean Inequalities for \hmath$q$-Partitioning Valuations}
We first begin with a generalization of \cref{thm:qparttailspecial}, which has the same proof. 
\begin{theorem}
\label{thm:qparttail}
Suppose that $v$ is a $q$-partitioning valuation over $[m],$ and $S\subseteq [m]$ is a random set in which each element appears independently. Then the following inequality holds for any $ a \ge 0, k\ge 0,$ $s,r\in \mathbb{N}$ such that $1\le s < r\le \log_2 q,$ and $\alpha \ge \frac{1}{s}:$  
    $$
    \prob[v(S)\ge \frac{r}{s}a+k]\le 
    t(\alpha, r, s)^{-k}\prob[v(S)\le a]^{-\alpha r}.
    $$
    In particular, choosing $a$ to be the median, $\alpha = \frac{1}{s}, t(\alpha, r, s) = \frac{r}{s},$ we recover \cref{thm:qparttailspecial}.
\end{theorem}

\noindent
Note that in the proof of \cref{thm:qparttailspecial}, we only needed $r\le \log_2 q$ to ensure that the sets $M_1, M_2, \ldots, M_r$ ``split'' $[m]$ into at most $q$ parts. This assumption, however, is unnecessary if $v$ is XOS (or $q = m$) as one cannot ``split'' $[m]$ into more than $m$ parts. This gives rise to even more ``fine-grained'' inequalities in the XOS case. To present them, however, we will make a slight change of notation. Before, we defined $t(\alpha, r,s)$ as the larger root of $t + \alpha r t^{-\frac{1}{\alpha s}} = \alpha r + 1.$
We can take $r$ and $s$ to be arbitrary integers satisfying $r>s$ when $v$ is XOS as discussed. Thus, the ratio $\frac{r}{s}$ can approximate an arbitrary real number $1+ \delta > 1.$ With this in mind, we make the following twist. Denote by $\xi(\psi, \delta)$ the larger root of the equation $\xi + \psi \xi^{-\frac{1+\delta}{\psi}} = \psi+1$ for some $\psi \ge 1+\delta >1.$ This is essentially the same equation after the substitution $\psi =\alpha r, 1+ \delta = \frac{r}{s}.$ The condition $\alpha \ge \frac{1}{s}$ is equivalent to $\psi \ge 1+\delta.$
We have:

\begin{theorem}
\label{thm:tailboundxos}
    Suppose that $v$ is a valuation over $[m]$ that is $\beta$-close (in the sense of \cref{def:closeness}) to being XOS and $S\subseteq [m]$ is a random set in which each element appears independently. Then the following inequality holds for any real numbers $\psi\ge 1+\delta>1,a>0, k\ge 0$
    $$
    \prob[v(S)\ge \frac{(1+\delta)}{\beta}a+k]\le 
    \xi(\psi, \delta)^{-k}\prob[v(S)\le a]^{-\psi}.
    $$
    In particular, choosing $a$ to be the median, $\psi = 1+\delta,\xi = 1+\delta$ the inequality becomes
    $$
    \prob[v(S)\ge \frac{(1+\delta)}{\beta}a+k]\le 
    \left(1+\delta\right)^{-k}2^{1+\delta}.
    $$
\end{theorem}
\begin{proof}
The proof is analogous to the one of \cref{thm:qparttail}, except that this time we have 
$$
\beta v(T\backslash M)\le \sum_{i=1}^r \frac{1}{s}v(M_i)\le 
\sum_{i=1}^r \frac{1}{s}v(T_i) = 
\sum_{i=1}^r \frac{1}{s}v(y^i)
\le 
\frac{ra}{s}
$$
as $v$ is $\beta$-close to being XOS.
\end{proof}

\noindent
We end with a discussion of median-mean inequalities, which are part of the motivation for our endeavour in this section. Namely, in \cite{RubinsteinW18}, the authors use the following crucial property of 1-Lipschitz subadditive valuations. Using Schechtman's bound (see \cref{section:introconcentration}), they obtain
$\expect[v(S)]\le 
3\median[v(S)] + O(1)
.$ We generalize this as follows.

\begin{proposition}
\label{prop:mediantomean}
Suppose that the non-negative random variable $Z$ satisfies the following inequality for some $0<\delta\le 1,$ and any $k >0.$
$$
\prob[Z\ge (1 + \delta)\median[Z] + k]\le 
(1+\delta)^{-k}2^{1+\delta}.
$$
Then, $\expect[Z]\le (1 + \delta)\median[Z] + O(\frac{1}{\delta}).$
\end{proposition}
\begin{proof}
Let $k$ be some non-negative real number that we will chose later. Then, 
$$
\expect[Z] = 
\int_0^{+\infty}
\prob[Z\ge t]\le 
(1 + \delta)\median[Z] + k
+
\int_{k}^{+\infty}
\prob[Z\ge (1 + \delta)\median[Z] + t]d t\le 
$$
$$
(1 + \delta)\median[Z] + k
+
\int_{k}^{+\infty}
2^{1+\delta}(1+\delta)^{-t}dt = 
(1 + \delta)\median[Z] + k
+
2^{1+\delta}\frac{1}{\ln (1 + \delta)}(1+\delta)^{-k}.
$$
Choosing $k = \frac{1}{\ln(1+\delta)}$ and using the inequality 
$0 <\delta \le 1,$ which also implies 
$\frac{\delta}{2}\le \frac{1}{\ln (1+\delta)}\le \delta,$ we bound the above expression by 
$$
(1 + \delta)\median[Z] + O\left(\frac{1}{\delta}\right).
$$
\end{proof}

\noindent
Applying this statement to $q$-partitioning valuations, we obtain the following two corollaries.

\begin{corollary}
\label{cor:medianmeanqpart}
If $v$ is $q$-partitioning, then 
$\expect[v(S)]\le (1 + O(\frac{1}{\log q}))\median[v(S)] + O(\log q).$
\end{corollary}
\begin{proof}
We just apply \cref{prop:mediantomean} for $\delta = \frac{1}{\lceil\log q\rceil}$ and combine with 
\cref{thm:qparttail}.
\end{proof}

\begin{corollary}
If $v$ is XOS, 
then $\expect[V(S)]\le \median[v(S)] + O(\sqrt{\median[v(S)]}).$
\end{corollary}
\begin{proof}
We apply \cref{prop:mediantomean} with $\delta = \frac{1}{\sqrt{\median[v(S)]}}$ and combine with 
\cref{thm:tailboundxos}.
\end{proof}

\noindent
Note that the last result matches the state-of-the-art median-mean bound for XOS function implied by the $\expect[v(S)]$-subgaussian behaviour of $v.$ Indeed, the lower tail-bound in \cref{thm:selfboundingqpart} for $v$ XOS
implies that $\median[v(S)]\ge \expect[v(S)] - 
O(\sqrt{\expect[v(S)]}).
$ More generally, median-mean bounds can also be derived using \cref{thm:selfboundingqpart}, but are only useful for large $q$ unlike \cref{cor:medianmeanqpart}. Namely,
 $\median[v(S)]\ge \expect[v(S)] - 
O(\sqrt{\frac{m}{q}\expect[v(S)]})
$ holds for $q$-partitioning valuations. 

\subsection{Technical Details}
\label{section:appendixtalphaqs}
We omitted the proof of \cref{lemma:talphaqs}. We give this proof here. First, however, we need to show that the equation $t + \alpha q t^{-\frac{1}{\alpha s}} = \alpha q + 1$ has two positive roots and one of them is larger than 1. Clearly, $t = 1$ is a root. Denote $f(t):= t + \alpha q t^{-\frac{1}{\alpha s}}-\alpha q - 1.$ Note that 
$f'(t) = 1 - \frac{q}{s}t^{-1-\frac{1}{\alpha s}}.$ As $q>s,$ $f'$ has a single root $t_0 = \left(\frac{q}{s}\right)^\frac{\alpha s}{\alpha s + 1}$ and this root is larger than $1$ (but smaller than $\frac{q}{s}$). Thus, $f'$ is decreasing in $(0,t_0) $ and increasing in $[t_0, +\infty).$ Therefore, $f(t_0)<f(1) = 0.$ Since $\lim_{t\longrightarrow+\infty} f(t) =+\infty,$ $f$ has just one more root and this root is larger than $1.$ We can now proceed to the proof of \cref{lemma:talphaqs}.

\noindent
\begin{proof}[Proof of \cref{lemma:talphaqs}]
Without loss of generality, let $0\le x_1 \le x_2\le \ldots\le x_q \le 1.$ Thus, we need to prove that 
$$
\min
\left(
t(\alpha, q, s),
(x_{q-s+1}x_{q-s+2}\cdots x_q)^{-\alpha}
\right) + 
\alpha \sum_{i = 1}^q x_i \le \alpha q + 1.
$$
Note that we can assume that $(x_{q-s+1}x_{q-s+2}\cdots x_q)^{-\alpha} \le t(\alpha,q,s).$ Indeed, otherwise we can increase the numbers $x_{q-s+1},x_{q-s+2},\ldots ,x_q$ until this inequality is satisfied and the left hand-side will only increase. Now, on, we will assume that $B = (x_{q-s+1}x_{q-s+2}\cdots x_q)^{-\alpha} \le t(\alpha,q,s).$\\

\noindent
Similarly note that we can assume that $x_1 = x_2 = \cdots = x_{q-s+1}.$\\

\noindent
Now, keeping $x_{q-s+1}$ fixed and $x_{q-s+1}x_{q-s+2}\cdots x_q$ fixed, note that the sum $\sum_{i = q-s+1}^qx_i$ is maximized when there exists some $0\le r\le s$ such that 
$$
1 = x_q = x_{q-1} = \cdots = x_{q-r+1}\ge 
x_{q-r}\ge x_{q-r-1} = x_{g-r-2} = \cdots = x_{q-s+1}.
$$
This is indeed the case since when we keep the product of two numbers $a\le b$ fixed, their sum increases as they get further apart. Formally, $a\gamma + b\gamma^{-1}\ge a + b$ for any $0<\gamma<1.$

\noindent
Under these assumptions, denote $y = x_{q-r}$ and
$x_{q-r-1} = x_{g-r-2} = \cdots = x_{q-s+1}   = x.$ Using this notation, we want to maximize 
$$
h(x,y) = (yx^{s-r-1})^{-\alpha} + \alpha(r + y + (s-r-1)x)
$$
in the set $\mathcal{K}=\{(x,y)>0 \; :\;(yx^{s-r-1})^{-\alpha} \le t(\alpha, q, s)\text{ and }0\le x\le y \le 1\}.$ Note that $\mathcal{K}$ is compact and $h$ is continuous. Therefore, there exists a maximizer. We can easily see that this maximizer is not in the interior of $\mathcal{K}$ as $\nabla_y h =  \alpha - \alpha y^{-\alpha - 1}x^{-\alpha (s-r-1)} <0$ in the interior. Thus, the gradient is non-zero and by moving in the direction of the gradient, the value of $h$ will increase. Thus, all the maximizers are on the boundary. There are three cases to consider:\\
\textbf{Case 1)} $y = x.$ Then, we need to prove that 
$$
x^{-\alpha (s-r)} + \alpha (r + (q-r)x)\le \alpha q + 1
$$
whenever $0\le r \le s$ and $x^{-\alpha (s-r)}\le t(\alpha, q, s).$ First, note that if $s = r,$ the inequality is trivial. For that reason, we assume that $0\le r <s-1$ now on. Consider the function $$g(x) = x^{-\alpha (s-r)} + \alpha (r + (q-r)x) - \alpha q - 1.$$ Then, 
$$
g'(x) = 
-\alpha(s-r)x^{-\alpha(s-r)-1} + (q-r),
$$
so the function $g'$ has  exactly one positive root $x_0.$ It follows that $g$ is decreasing in $(0,x_0)$ and increasing in $(x_0,+\infty).$ No matter what the value of $x_0$ is, this means that the maximal value of $g$ on $[t(\alpha, q, s)^{-\frac{1}{\alpha(s-r)}},1]$ (which is the set of feasible values for $x$) is always at either the point $1$ or the point $t(\alpha, q, s)^{-\frac{1}{\alpha(s-r)}}.$ Thus, we simply need to prove that $g(1)\le 0$ and 
$g(t(\alpha, q, s)^{-\frac{1}{\alpha(s-r)}})\le 0.$ The first inequality is trivial. The second is equivalent to:
$$
t(\alpha, q, s) + 
\alpha r  + \alpha(q-r)t(\alpha, q, s)^{-\frac{1}{\alpha(s-r)}}\le \alpha q + 1 \Longleftrightarrow
$$
$$
t(\alpha,q,s)^{-\frac{1}{\alpha(s-r)}}
\le \frac{\alpha (q-r)+1 - t(\alpha,q,s)}{\alpha (q-r)}.
$$
Note that for this inequality to hold, we first need to prove that $\alpha (q-r)+1 - t(\alpha,q,s)\ge 0.$ In fact, we will prove that 
$
\alpha (q-s)+1 \ge  t(\alpha,q,s).$ Since $\alpha \ge \frac{1}{s},$ we see that $\alpha (q-s)+1 \ge \frac{q}{s}\ge t_0,$ where $t_0$ is the root of the derivative of $f(t) = t + \alpha qt^{-\frac{1}{\alpha s}},$ defined in the beginning of this appendix. Since $f$ is increasing on $[t_0,+\infty),$ all we need to prove is that 
$$
f(\alpha (q-s)+1)\ge 
f(t({\alpha,q,s})) = 0 \Longleftrightarrow
$$
$$
(1 + \alpha(q-s)) + \alpha q (1+\alpha (q-s))^{-\frac{1}{\alpha s}}\ge \alpha q + 1 \Longleftrightarrow
$$
$$
\alpha q (1+\alpha (q-s))^{-\frac{1}{\alpha s}} \ge \alpha s \Longleftrightarrow
$$
$$
(1+\alpha (q-s))^{-\frac{1}{\alpha s}} \ge \frac{s}{q} \Longleftrightarrow
$$
$$
(1+\alpha (q-s))^{\frac{1}{\alpha s}} \le \frac{q}{s}.
$$
However, as $\frac{1}{\alpha s}\le 1$ by the choice of $\alpha, $ by the famous Bernoulli inequality, we know that 
$$
(1+\alpha (q-s))^{\frac{1}{\alpha s}}\le 
1 + \frac{\alpha (q-s)}{\alpha s} = \frac{q}{s},
$$
as desired.\\

\noindent
Now that we know $\frac{\alpha (q-r)+1 - t(\alpha,q,s)}{\alpha (q-r)}\ge 0,$ the desired inequality becomes equivalent to proving 
$$
t(\alpha,q,s)^{-\frac{1}{\alpha s}}
\le \left(\frac{\alpha (q-r)+1 - t(\alpha,q,s)}{\alpha (q-r)}\right)^\frac{s-r}{s}.
$$
Using that $t(\alpha,q,s)^{-\frac{1}{\alpha s}} = \frac{\alpha q + 1 - t(\alpha,q,s)}{\alpha q},$ the above inequality becomes equivalent to 
$$
\frac{\alpha q + 1 - t(\alpha,q,s)}{\alpha q}\le 
 \left(\frac{\alpha (q-r)+1 - t(\alpha,q,s)}{\alpha (q-r)}\right)^\frac{s-r}{s} \Longleftrightarrow
$$
$$
\left(
\frac{\alpha q + 1 - t(\alpha,q,s)}{\alpha q}
\right)^{\alpha s}\le 
 \left(\frac{\alpha (q-r)+1 - t(\alpha,q,s)}{\alpha (q-r)}\right)^{\alpha (s-r)}.
$$
To prove this inequality, denote by $b = \alpha(q-s)+1- t(\alpha,q,s)\ge 0, a = t(\alpha,q,s) - 1> 0.$ Then, we want to prove that 
$$
\left(
\frac{b + \alpha s}{a + b + \alpha s}
\right)^{\alpha s}\le 
\left(
\frac{b + \alpha (s-r)}{a + b + \alpha (s-r)}
\right)^{\alpha (s-r)}.
$$
To prove this, we simply show that the function $$
x\longrightarrow
\left(
\frac{b + x}{a + b + x}
\right)^{x}
$$
is decreasing. Equivalently, we want to show that its logarithm $k(x)  = x\ln \left(
\frac{b + x}{a + b + x}
\right)$ is decreasing. This, however, is simple as 
$$
k'(x) = 
\ln \left(
\frac{b + x}{a + b + x}
\right) + 
x \frac{a+b+x}{b+x}\frac{a}{(a+b+x)^2} = 
$$
$$
\ln \left(
\frac{b + x}{a + b + x}
\right) + 
\frac{x}{b+x}\frac{a}{a+b+x}\le 
$$
$$
\ln \left(
\frac{b + x}{a + b + x}
\right) + 
\frac{a}{a+b+x} = 
\ln \left(
\frac{b + x}{a + b + x}
\right) + 1 - 
\left(
\frac{b + x}{a + b + x}
\right)\le 0,
$$
as $\ln (y) + 1 - y\le 0 $ for all $y.$ With this, case 1 is complete.\\
\textbf{Case 2)} $y = 1.$ This is the same case as $1$ except that we replace $r$ with $r-1.$\\
\textbf{Case 3)} $(yx^{s-r-1})^{-\alpha} = t(\alpha, q, s).$ Then, we know that $y = t(\alpha, q, s)^{-\frac{1}{\alpha}}x^{-(s-r-1)}.$ We need to prove that 
$$
t(\alpha, q, s) + \alpha r + \alpha (q-r-1)x + \alpha t(\alpha, q, s)^{-\frac{1}{\alpha}}x^{-(s-r-1)}\le 
\alpha q + 1.
$$
Consider the function $m(x) = \alpha (q-r-1)x + \alpha t(\alpha, q, s)^{-\frac{1}{\alpha}}x^{-(s-r-1)}.$ Then
$$
m'(x) = 
\alpha (q-r-1) - 
\alpha t(\alpha, q, s)^{-\frac{1}{\alpha}}(s-r-1)x^{-(s-r)}.
$$
Note that this function has a unique root
$x_0.$ Therefore, $m$ is decreasing on $(0,x_0]$ and increasing on $[x_0, +\infty).$ In other words, the maximal values of $m(x)$ on $[t(\alpha,q,s)^{-\frac{1}{\alpha(s-r-1)}},t(\alpha,q,s)^{-\frac{1}{\alpha(s-r)}}]$ (which is the feasible set for $x$ as $x^{-\alpha(s-r-1)}\le (yx^{s-r-1})^{-\alpha}  = t(\alpha,q,s)$ and\linebreak $x^{-\alpha(s-r)}\ge (yx^{s-r-1})^{-\alpha}  = t(\alpha,q,s)$)
 are at the two points $x = t(\alpha,q,s)^{-\frac{1}{\alpha(s-r-1)}}$ and \linebreak $x = t(\alpha,q,s)^{-\frac{1}{\alpha(s-r)}}.$ However, these cases correspond to $y = 1$ and $y = x,$ which were already analyzed in case 1 and case 2.\\
 
\noindent
Finally, we want to prove that the choice of $t$ is optimal. This follows simply by taking\linebreak $x_1 = x_2 = \cdots = x_q = t^{-\frac{1}{\alpha s}}.$ 
\end{proof}

\begin{remark}
\label{rmk:alphalessthansiso}
\normalfont

We end this appendix with the remark that one can also obtain similar inequalities when $\alpha <\frac{1}{s},$ but with a potentially different choice of $t.$ Namely, suppose that $\alpha >0$ and $t^{min}(\alpha ,q,s)$ is the smallest number among $t_0,t_1,\ldots, t_{s-1},$ where $t_r$ is the larger root of\linebreak $t + \alpha (q-r)t^{-\frac{1}{\alpha (s-r)}} = \alpha (q-r)+1.$ From the proof of \cref{lemma:talphaqs}, it follows that 
when $\alpha \ge\frac{1}{s},$ it is the case that $t^{min}(\alpha ,q,s) =t_0 = t(\alpha, q, s).$ This, however, might not be the case in general. For example, we can compute that when $q = 5, s = 2, \alpha = \frac{1}{10},$
we have $t_0 \approx 1.41, t_1\approx 1.38,$ so $t(\alpha ,q,s) = t_0\neq \min(t_0,t_1) = t^{min}(\alpha,q,s).$\\
 
\noindent
With this in mind, the same proofs show that \cref{lemma:talphaqs} holds with the value $t^{min}(\alpha,q,s)$ instead of $t(\alpha,q,s)$ for any $\alpha >0.$ Similarly,
\cref{thm:talagrandgenerals} holds 
with the value $t^{min}(\alpha,q,s)$ instead of $t(\alpha,q,s)$ for any $\alpha >0.$ We did not state the result in this more general form earlier as we specifically wanted to derive the inequalities for $\alpha = \frac{1}{s}, t(\alpha,q,s) = \frac{q}{s}.$ 
\end{remark}
\section{Distance to Subadditive Valuations}
\label{section:closeness}
We will show that subadditive valuations over $[m]$ are $\Omega(\frac{1}{\log q})$-close to $q$-partitioning valuations over $[m]$ and this factor is asymptotically tight. Both directions follow closely the framework for XOS 
approximations of subadditive valuations in \cite{BhawalkarR11}. Denote by $\mathcal{H}_a$ the $a$-th harmonic number, i.e. $\displaystyle\mathcal{H}_a = \sum_{i = 1}^a\frac{1}{i} = \ln a + O(1).$

\noindent
\begin{proposition}
\label{lem:qparttosa}
 $\mathcal{Q}(2,[m])$ is $\frac{1}{\mathcal{H}_{q-1}}$-close to $\classqm.$
\end{proposition}
\begin{proof} Take any subadditive valuation $g$ over $[m].$ Upon taking the dual of the linear program in \cref{def:closeness}, we need to prove the following fact. For any $S\subseteq [m]$ and any partition of $S$ into $q$ parts $S_1, S_2, \ldots, S_q,$ the optimal value of the following linear program 
\begin{equation*}
    \begin{split}
        \max \; &\sum_{j \in [q]} p_j \text{ s.t.,}\\
        &\sum_{j \in I}p_j\le g(\bigcup_{j \in I}S_j)\; \forall  I\subseteq [q],\\
        &p_j\ge 0 \; \forall j \in [q]
    \end{split}
\end{equation*}
is at least $g(S)/\mathcal{H}_{q-1}.$ We construct the price vector $(p_1, p_2, \ldots, p_q)$ explicitly with the following algorithm:
\begin{tcolorbox}[colback=black!5!white,colframe=black!75!black, title = {Algorithm for Constructing Price Vectors}]
\textbf{Initialize:} $C = \emptyset.$

\noindent
\textbf{Iteration:} While $C\neq [q]:$
\begin{itemize}
    \item Find $\displaystyle A \in  \arg \min_{\emptyset \subsetneq A'\subseteq [q]}\frac{g(\bigcup_{i\in A'}S_i)}{|A'\backslash C|}.$
    \item Set $p_j = \frac{g(\bigcup_{i\in A}S_i)}{|A\backslash C|\times \mathcal{H}_{q-1}}$ for all $j \in A.$
    \item Update $C = C\cup A.$
\end{itemize}

\noindent
\textbf{Return:} Output the price vector.  
\end{tcolorbox}

\noindent
We need to show two things. First, for any $I\subsetneq [q],$ it must be the case that $\sum_{j \in I}p_j\le g(\bigcup_{j \in I}S_j).$ Second, the fact that $\sum_{j \in [q]}p_j\ge g(S)/\mathcal{H}_{q-1}.$ Technically, we also need to show that the inequality $\sum_{j \in [q]}p_j\le g(S)$ holds. However, if this condition is violated, we can reduce prices until $\sum_{j \in [q]}p_j\le g(S)$ without violating any other conditions.\\

\noindent
First, take some $I\subsetneq [q].$ Consider the iteration of the algorithm in which the $\ell'$th price indexed by an element in $I$ is determined. Since the algorithm could have chosen $I$ in that iteration, it must be the case that  
$$
\frac{g(\bigcup_{i\in A}S_i)}{|A\backslash C|}\le 
\frac{g(\bigcup_{i\in I}S_i)}{|I\backslash C|} \le 
\frac{g(\bigcup_{i\in I}S_i)}{|I|-\ell+1}.
$$
Therefore, 
$$
\sum_{i\in I}p_i \le 
\sum_{\ell = 1}^{|I|}
\frac{g(\bigcup_{i\in I}S_i)}{\mathcal{H}_{q-1}(|I|-\ell+1)}\le 
\frac{g(\bigcup_{i\in I}S_i)}{\mathcal{H}_{q-1}}\sum_{\ell = 1}^{|I|}\frac{1}{\ell} \le
g(\bigcup_{i\in I}S_i),
$$
as desired (we used the fact that $|I|\le q-1$).
Now, we will analyze $\sum_{i\in [q]}p_i.$ Let the index sets chosen by the algorithm be $A_1, A_2, \ldots, A_t.$ Then,
$$
\sum_{i\in [q]}p_i = 
\sum_{j = 1}^t\sum_{i \in A_j}p_j = 
\sum_{j = 1}^t |A_j|\times 
\frac{g(\bigcup_{i\in A_j}S_i)}{\mathcal{H}_{q-1}\times |A_j|} = 
\frac{\sum_{j = 1}^tg(\bigcup_{i\in A_j}S_i)}{\mathcal{H}_{q-1}}\ge 
\frac{g(S)}{\mathcal{H}_{q-1}},
$$
where we used subadditivity in the last equality.
\end{proof}

\noindent
A simple modification of \cite[Appendix C]{BhawalkarR11} shows that \cref{lem:qparttosa} is tight up to a constant multiplicative factor.

\begin{proposition}
\label{lem:salogclose}
For any $2< q \le m,$ the class of subadditive valuations over $[m]$ is not $\gamma$-close to $\classqm$ for any
$\gamma >\frac{2}{\log_2 \frac{q}{2}}.$
\end{proposition}
\begin{proof}
Let $q'$ be the largest integer such that $q'\le q$ and $q' = 2^a - 1$ for some natural number $a.$ Note that $q' \ge \frac{q}{2}.$ Then, one can construct as in \cite[Appendix C]{BhawalkarR11} a subadditive valuation $g$ over $[q']$ that is not $\gamma$-close to being XOS for any $\gamma >\frac{2}{\log_2 \frac{q}{2}}$ as follows.\\

\noindent
\textit{Construction in \cite[Appendix C]{BhawalkarR11}}: 
Identify $[q']$ with the set $\mathcal{V}$ of $2^a-1$ non-zero vectors over the finite vector space $F_2^a.$ For each $v\in \mathcal{V},$ set $S_v= \{u \in \mathcal{V}\; : v\cdot u \equiv 1 \pmod{2}\}.$ Define $g$ as the set-cover function over $\mathcal{V}.$ On the one hand, we can observe that $g(\mathcal{V})\ge a.$ Indeed, for any $r<a$ vectors $v_1, v_2, \ldots, v_r,$ we can find some $u$ such that $v_i \cdot u = 0\pmod{2}$ holds for all $i\in[r]$ simply because the matrix $(v_1\; v_2\;\cdots\; v_r)^T$ is not of full rank. On the other hand, note that for each $v,$ the set $S_v$ contains $2^{a-1}$ elements and each $u\in \mathcal{V}$ belongs to $2^{a-1}$ sets of the form $S_v.$ Since $g(S_v) = 1,$ the fractional cover $\alpha$ assigning weight $\frac{1}{2^{a-1}}$ to each set $S_v$ satisfies 
$$
\sum_{I\subseteq \mathcal{V}} \alpha(I)g(I) = \sum_{v\in \mathcal{V}}
\alpha(S_v)g(S_v) = 
(2^{a}-1)\times \frac{1}{2^{a-1}}\le 2,
$$
which shows that $g$ is not $\gamma$-close to being XOS for $\gamma
>\frac{2}{a}\ge \frac{2}{\log_2 \frac{q}{2}}.$\\ 

\noindent
Now, we go back to the problem statement.
Clearly, $q'\le q\le m.$ First, we extend $g$ to a subadditive valuation $g'$ on $[m]$ by setting $g'(S):= g(S\cap [q'])$ for any $S\subseteq [m].$ It trivially follows that $g'$ is not $\gamma$-close to any $q'$-partitioning valuation for any $\gamma >\frac{2}{\log_2 \frac{q}{2}}'.$ As $q'\le q,$ meaning that any  $q$-partitioning valuation is also $q'$-partitioning, the result follows.
\end{proof}

\end{document}